\pdfoutput=1 

\documentclass[aip, reprint, floats, floatfix]{revtex4-1}
\usepackage[final]{epsfig}
\usepackage{amsmath, amsthm, amssymb}

\usepackage{graphicx}
\usepackage{color}
\usepackage{ctable}
\usepackage{booktabs}
\usepackage{mathtools}
\usepackage{algpseudocode}
\usepackage{setspace}
\RequirePackage{lettrine} 

\DeclarePairedDelimiter{\floor}{\lfloor}{\rfloor}

\newcommand{\dropcapfont}{\fontfamily{lmss}\bfseries\fontsize{26pt}{28pt}\selectfont}
\newcommand{\dropcap}[1]{\lettrine[lines=2,lraise=0.05,findent=0.1em, nindent=0em]{{\dropcapfont{#1}}}{}}

\begin{document}

\title[]{
Sizing the length of complex networks
} 

\author{Gorka Zamora-L\'opez}		\email{gorka@Zamora-Lopez.xyz}
\affiliation{Department of Information and Communication Technologies, Universitat Pompeu Fabra, Barcelona, Spain.}
\affiliation{Center for Brain and Cognition, Universitat Pompeu Fabra, Barcelona, Spain.}

\author{Romain Brasselet}
\affiliation{Scuola Internazionale Superiore di Studi Avanzati (SISSA), Trieste, Italy.}


\begin{abstract}

\noindent
Among all characteristics exhibited by natural and man-made networks the small-world phenomenon is surely the most relevant and popular. But despite its significance, a reliable and comparable quantification of the question ``how small is a small-world network and how does it compare to others'' has remained a difficult challenge to answer. 
Here we establish a new synoptic representation that allows for a complete and accurate interpretation of the pathlength (and efficiency) of complex networks. We frame every network individually, based on how its length deviates from the shortest and the longest values it could possibly take. For that, we first had to uncover the upper and the lower limits for the pathlength and efficiency, which indeed depend on the specific number of nodes and links. These limits are given by families of singular configurations that we name as ultra-short and ultra-long networks.
The representation here introduced frees network comparison from the need to rely on the choice of reference graph models (e.g., random graphs and ring lattices), a common practice that is prone to yield biased interpretations as we show. 
Application to empirical examples of three categories (neural, social and transportation) evidences that, while most real networks display a pathlength comparable to that of random graphs, when contrasted against the absolute boundaries, only the cortical connectomes prove to be ultra-short. 
\end{abstract}

\maketitle

\clearpage
\newpage
\mbox{~}
\clearpage
\newpage

\dropcap{T}he small-world phenomenon has fascinated popular culture and science for decades. Discovered in the realm of social sciences during the 1960s, it arises from the observation that any two persons in the world are connected through a short chain of social ties~\cite{Milgram_SW_1967}.
Since then many real networks have been found to exhibit the small-world phenomenon as well~\cite{Watts_WSmodel_1998, Newman_Review_2003, Boccaletti_Review_2006}, from natural to man-made systems. But, how small is a small-world network and how does it compare to other networks? The small-world phenomenon relies on the computation of the average pathlength -- the average distance between all pairs of nodes. 
Since the average pathlength very much depends on the number of nodes and links, comparing it across networks is a non-trivial task. Therefore, in general, when we say that ``\emph{a complex network is small-world}'' we mean, without  further quantitative accuracy, that its average pathlength is much smaller than the number of nodes it is made of~\cite{Newman_Review_2003}. 

Consider two empirical networks. $G_1$ is a small social network, e.g., a local sports club, of $N_1 = 100$ members. A link between two members implies they trust each other. $G_2$ is an online social network with a million users ($N_2 = 10^6$) where two profiles are connected if both users are friends with each other. When comparing these two systems, even if we found that the average pathlength $l_1$ of $G_1$ is smaller than the length $l_2$ of $G_2$, we could not conclude that the internal topology of the local sportsclub is shorter, or more efficient, than the structure of the large online network. The observation $l_1 < l_2$ could be a trivial consequence of the fact that $N_1 \ll N_2$.
In order to fully interpret the length or efficiency of complex networks we need to disentangle the contribution of the network's internal architecture to the pathlength from the incidental influence contributed by the number of nodes and links.

The usual strategy to deal with this problem in practice has been to compare empirical networks to well-known graph models: random graphs and regular lattices~\cite{Watts_WSmodel_1998, Humphries_SWness_2008, Zamora_PathsCat_2009, Muldoon_SWpropensity_2016, Basset_SmallWorld_2016}. These models represent a variety of null-hypotheses, useful to answer particular questions we may have about the data. However, they do not correspond to absolute boundaries of the pathlength or efficiency of complex networks~\cite{Barmpoutis_Extremal_2010, Barmpoutis_Extremal_2011, Gulyas_Pathlength_2011}. For example, scale-free networks are known to be smaller than random graphs~\cite{Cohen_UltraSmall_2003}. As a consequence, their use as references may give rise to biased interpretations.

Here we establish a reference framework under which the average pathlength and efficiency~\cite{Latora_Efficiency_2001} of networks can be interpreted and compared. Instead of relying on the comparison to typical models, we evaluate how the length and efficiency of a network - of a given size and density - deviate from the smallest and the largest values they could possibly take. To do so, we first uncover the upper and the lower limits for the pathlength and efficiency of networks, which indeed depend on the specific number of nodes and links. We find that these limits are given by families of singular configurations we will refer to as \emph{ultra-short} and \emph{ultra-long} networks. With these boundaries at hand, we show that typical models (random, scale-free and ring networks) undergo a transition as their density increases, eventually becoming ultra-short. The convergence rate, however, depends on the properties of each model. Finally, we study a sample set of well-known empirical networks (neural, social and transportation). While most of these graphs display a pathlength close to that of random graphs, when contrasted against the absolute boundaries only the cortical connectomes prove to be quasi-optimal.

\section*{Results}

\begin{figure*} 
	\centering
	\includegraphics[width=1.0\textwidth,clip=]{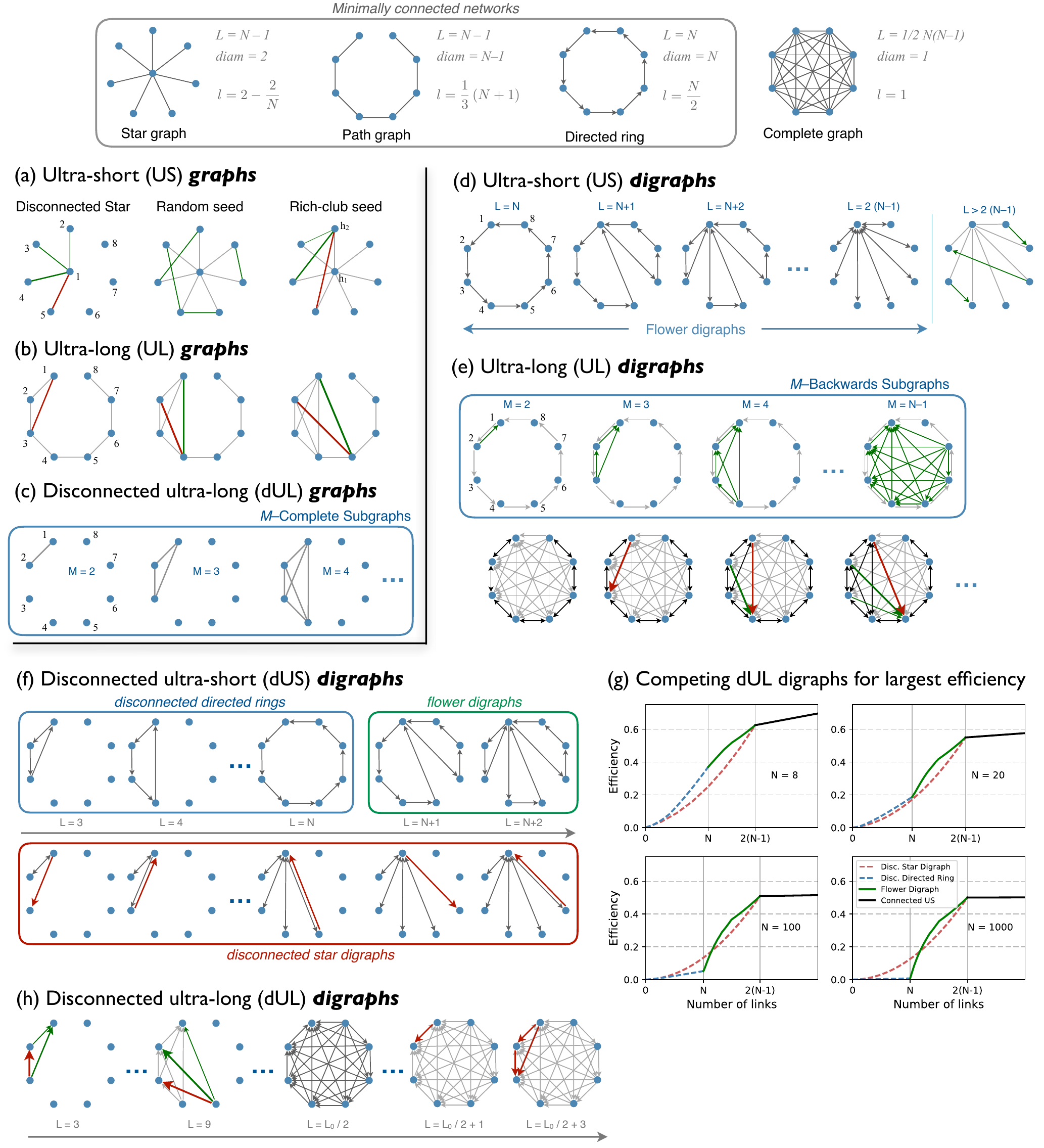}
	\caption{ 		\label{fig:USULnets}
	{\bf Construction of ultra-short and ultra-long networks.} 
	(Top) Sparsest and densest connected networks. These well-known networks serve as the starting points to construct extremal graphs and digraphs of arbitrary density. (a-c) Procedures to build ultra-short and ultra-long graphs, both connected and disconnected. Edge colour denotes the order of edge addition. Red edges are the last added and green links the ones in the previous steps. (d-h) Generation of directed graphs with extremal pathlength or efficiency. These cases are often non-Markovian and lead to novel structures. (d) In the sparse regime ultra-short digraphs are characterised by flower digraphs, i.e., a collection of directed cycles converging at a single hub. Every arc added leads to a flower digraph with an additional ``petal''. (e) Although several configurations may lead to digraphs with longest pathlength, we introduce here an algorithmic approximation to the upper bound, $M$-backwards subgraphs. (f) Up to three different digraph configurations compete for largest efficiency. (g) The winner depends on network density. Finally, (h) digraphs with smallest efficiency are achieved by constructing the densest directed acyclic graphs possible, i.e., minimising the contribution of cycles to the path structure of the network.
	} 
\end{figure*}

In order to avoid the ambiguous meaning of the term \emph{size} in the literature, in the following we will use \emph{size} only to refer to the number of nodes $N$ in a network and we will correspondingly employ the adjectives \emph{small} and \emph{large}. We will refer to the average pathlength $l$ of a network as its \emph{length} and use corresponding adjectives \emph{short} and \emph{long}. We will denote the properties of directed graphs adding a tilde to the symbols. For example, if $L$ is the number of undirected edges in a graph, $\tilde{L}$ will be the number of directed arcs in a directed graph (digraph).

\subsection*{Ultra-short and ultra-long networks}

Figure~\ref{fig:USULnets} summarises the families of directed and undirected graph configurations with the shortest and longest possible average pathlength, as well as the largest and smallest efficiency; see Methods and Supplementary Text for details. These families arise from a few simple building blocks, Fig.~\ref{fig:USULnets} (top). 
The sparsest connected graphs that can be constructed are named \emph{trees}, i.e., graphs without cycles. All trees of size $N$ contains $L = N-1$ edges. Among them, star and path graphs are the ones with the shortes and the longest pathlength respectively.
In a star graph, any two nodes can reach each other jumping through the hub while in a path graph, the whole network needs to be traversed to travel from one end to the other. In case the links are directed, however, directed rings are the sparsest connected digraphs, which consist of $\tilde{L} = N$ arcs, all pointing in the same orientation. Finally, a complete graph is the network in which all nodes are connected to each other, thus containing $L_o = \frac{1}{2} N(N-1)$ edges or $\tilde{L}_o = N(N-1)$ directed arcs. The average pathlength of a complete graph is $l_o = 1$, the shortest of all networks.

Ultra-short and ultra-long graphs of arbitrary $L$ can be achieved by adding edges to star and path graph respectively. In the case of digraphs, both ultra-short and ultra-long configurations are obtained by adding arcs to a directed ring. The precise order of link addition differs from case to case. Among many findings two are of special mention. ($i$) Ultra-short and ultra-long graphs can be generated adding edges one-by-one to the initial configurations, see Figs.~\ref{fig:USULnets}(a) and (b), but construction of extremal digraphs is often non-Markovian. That is, an ultra-short or an ultra-long digraph with $\tilde{L} + 1$ arcs cannot always be achieved by adding one arc to the extremal digraph with $\tilde{L}$ arcs. For example, Fig.~\ref{fig:USULnets}(d) shows that the digraphs with shortest pathlength initially transition from a directed ring onto a star graph following unique configurations we named \emph{flower digraphs}. ($ii$) All networks of a given size and density with diameter $diam(G) = 2$ have exactly the same pathlength and they are ultra-short, regardless of how their links are internally arranged. See the \emph{ultra-short network theorem} in Supplementary Text. 

When studying large networks it is common to find that these are sparse and fragmented into many components. While the pathlength of such networks is infinite, these cases can still be characterised by their efficiency, which remains a finite quantity allowing to ``zoom-in'' into the sparse regime. We remind that the efficiency of a network is defined as the average of the inverses of the pairwise distances. Thus the contribution of disconnected pairs (with infinite distance) vanishes. We could identify sparse configurations [$L < N-1$ or $\tilde{L} < 2(N-1)$] with the largest efficiency, whose efficiency transitions from zero (for an empty network) to that of a star graph. In the case of graphs, Fig.~\ref{fig:USULnets}(a) there is a unique optimal configuration but for digraphs we found that up to three different structures compete for the largest efficiency when $\tilde{L} < 2(N-1)$, Figs.~\ref{fig:USULnets}(f) and (g). On the other hand, the least efficient network is always disconnected. Therefore, for any connected network, there is always a disconnected one with the same number of nodes and links, and with smaller efficiency. See Figs.~\ref{fig:USULnets}(c) and (h) for the graph and digraph configurations with smallest efficiency possible. The efficiency of such (di)graphs equals the density of links.

\subsection*{The length of common network models}		\label{sec:Models}

In the following we illustrate how the ultra-short and ultra-long boundaries frame the space of lengths that networks can possibly take. We start by investigating the null-models which over the years have dominated the discussions on the topic of small-world networks: random graphs, scale-free networks and ring lattices. We consider undirected and directed versions with $N=1000$ nodes and study the whole range of densities; from empty ($\rho = 0$) to complete ($\rho = 1$). The results are shown in Figure~\ref{fig:NetModels}. Shaded areas mark the values of pathlength and efficiency that no network can achieve. Solid lines represent the ranges in which the models are connected and dashed lines correspond to the efficiencies of disconnected networks. The location of the original building-blocks (star graphs, path graphs, directed rings and complete graphs) are also represented over the maps for reference.

The pathlength of random, scale-free and ring networks decays with density, as expected, with the three cases eventually converging onto the lower boundary and becoming ultra-short. But, the decay rates differ for each model. Scale-free networks are always shorter than random graphs in the sparser regime, Fig.~\ref{fig:NetModels}(b), where the length of both models is well above the lower boundary. However, the two models converge simultaneously onto the ultra-short limit at $\rho \approx 0.08$. On the other hand, the ring lattices decay much slower and only becomes ultra-short at $\rho \approx 0.5$.

\begin{figure*}
	\centering
	\includegraphics[width=1.0\textwidth,clip=]{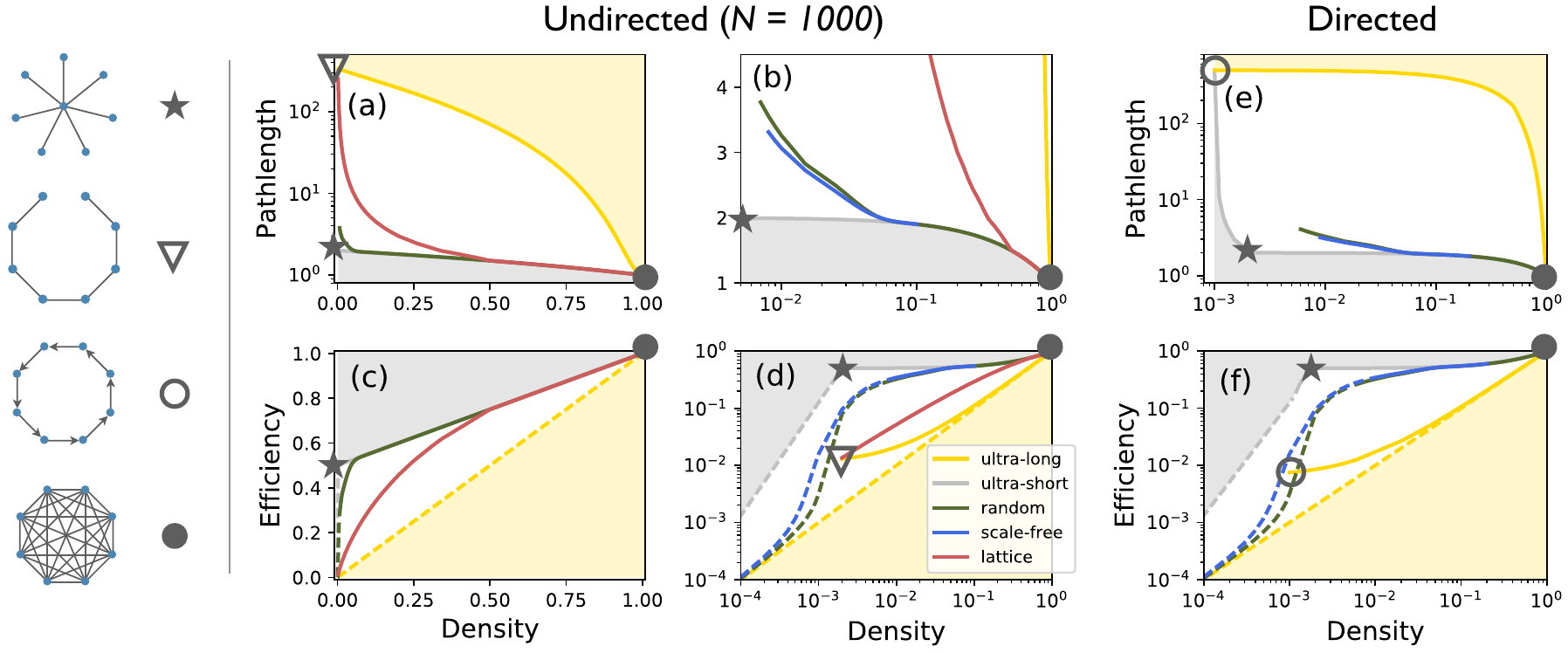}
	\caption{ 	\label{fig:NetModels}
	{\bf Pathlength of characteristic network models.}
	(a) and (b): Average pathlength of ring lattices (red), random (green) and scale-free (blue) graphs of $N = 1000$ nodes, compared with corresponding upper and lower boundaries for ultra-long (yellow) and ultra short (grey) graphs. Shaded areas mark values of pathlength that no graph of the same size can achieve depending on density. The pathlength of the three models decay towards the ultra-short boundary at sufficiently large density.
	(c) and (d): Same for the efficiency of networks. The lower boundary (ultra-long) is represented by two lines: a dashed line representing disconnected graphs $E_{dUL} \approx \rho$ and a solid line for connected graphs. The efficiency of random and scale-free graphs undergoes a transition from ultra-long to ultra-short centred at their percolation thresholds.
	(e) Patlength of random and scale-free digraphs. In this case, the two boundaries emerge from the same point corresponding to a directed ring (red cross). (f) Efficiency of random and SF digraphs.
	Curves for random and scale-free networks are averages over 1000 realisations. Dashed lines represent ranges of density for which the models are disconnected and solid lines represent (di)graphs which are connected. 
	} 
\end{figure*}

Figures~\ref{fig:NetModels}(c) and (d) reproduce the same results in terms of efficiency. An advantage of efficiency is that it always takes a finite value, from zero to one, regardless of whether a network is connected or not. Zooming into the sparser regime, we observe that the efficiency of both random ($E_r$) and scale-free ($E_{SF}$) graphs undergoes a transition, shifting from the ultra-long to the ultra-short boundary, Fig.~\ref{fig:NetModels}(d). They are nearly identical except for a narrow regime in between $\rho \in (4\! \times\! 10^{-4}, 2 \!\times\! 10^{-3}$). Here, $E_{SF}$ grows earlier than $E_r$, reaching a peak difference of $E_{SF} \approx  5 \times E_{r}$ at $\rho = 10^{-3}$. The reason for this is that SF graphs percolate earlier than random graphs~\cite{Cohen_Internet_2000}. Indeed, the onset of a giant component in random graphs of size $N = 1000$ happens at $\rho \approx 10^{-3}$. 

The results for the directed versions of the random and scale-free networks, Figs.~\ref{fig:NetModels}(e) and (f), are very similar. The main difference is that when $\tilde{L} = N$ both the upper and the lower boundaries are born from the same point, which corresponds to the initial directed ring, panel~(e).

\subsection*{Interpretation and comparison of empirical networks}			\label{sec:Empirical}

We now illustrate how knowledge of the true boundaries allows to quantify and interpret the length of real networks faithfully. Given two networks $G_1$ and $G_2$ with pathlengths $l_1 < l_2$, we could claim that $G_1$ is shorter than $G_2$. But, if $G_2$ is bigger, i.e. $N_1 < N_2$, then the fact that $l_1 < l_2$ does not necessarily imply that the \emph{topology} of $G_1$ is more efficient than the \emph{topology} of $G_2$. In order to clarify this we may normalise their pathlengths and define the following relative measures $l'_1 = l_1 \, / \, N_1$ and $l'_2 = l_2 \, / \, N_2$. The shortest topology should then correspond to the network with shorter $l'$. This conclusion, however, would only be fully informative if the link densities of both networks were the same.

Random graphs and ring lattices have been often employed as the references to characterise the ``\emph{small-worldness}'' of complex networks. Sometimes the relative pathlength $l' = l / \, l_r$ is defined which considers the length $l_r$ random graphs as the lower boundary~\cite{Humphries_SWness_2008}. This measure takes $l' = 1$ when the length of the real network matches that of random graphs. In other cases a 2-point normalisation has been proposed which considers also ring lattices as the upper boundary~\cite{Zamora_PathsCat_2009, Muldoon_SWpropensity_2016}, and a two-point normalisation is used $l' = (l - l_r) \, / \, (l_{latt} - l_r)$. In this case $l' = 0$ if the length of the real network equals that of random graphs (the lower boundary) and $l' = 1$ if it matches the length of ring lattices (the upper boundary). Using the actual ultra-short and ultra-long boundaries we have identified, we can redefine the 1-point and 2-point normalisations as:
\begin{eqnarray}	
	l' & = & \frac{l}{l_{US}}				\label{eq:1point} 		\\
	l' & = & \frac{l - l_{US}} {l_{UL}-l_{US}}.	\label{eq:2point}
\end{eqnarray}
For practical illustration, we study a set of empirical networks from three different domains: neural and cortical connectomes, social networks and transportation systems, see Table~\ref{tab:RealNets}. These examples represent a diverse set of real networks with sizes ranging from $N = 34$ to $4941$ and densities from $\rho \approx 10^{-4}$ to $0.330$. The results are shown in Figure~\ref{fig:Pathlen_RealNets}. The absolute pathlengths in panel (a) reveal that cortical and neural connectomes are shorter than social and transportation networks. Now, we want to understand whether this observation is a trivial consequence of the different sizes and densities of those networks. First, we apply the normalisation $l' = l \, / \, N$. In this case, the ranking is very much altered, panel (b). The short length observed for the cortico-cortical connectomes seems to be partly explained by their small size ($N < 100$). The \emph{Caenorhabditis elegans}, which is the biggest of the four neural networks, is now the shortest of them in relative terms. Among the social networks, the Zachary karate club (which is the smallest network in the data set) becomes the ``longest'' network of all, while the three largest (Facebook circles, world-wide airport transportation and the U.S.A. power grid) become the ``shortest''. The network of prison inmates is directed and weakly connected, therefore it has an infinite pathlength.

\begin{figure*} 
	\centering
	\includegraphics[width=0.8\textwidth,clip=]{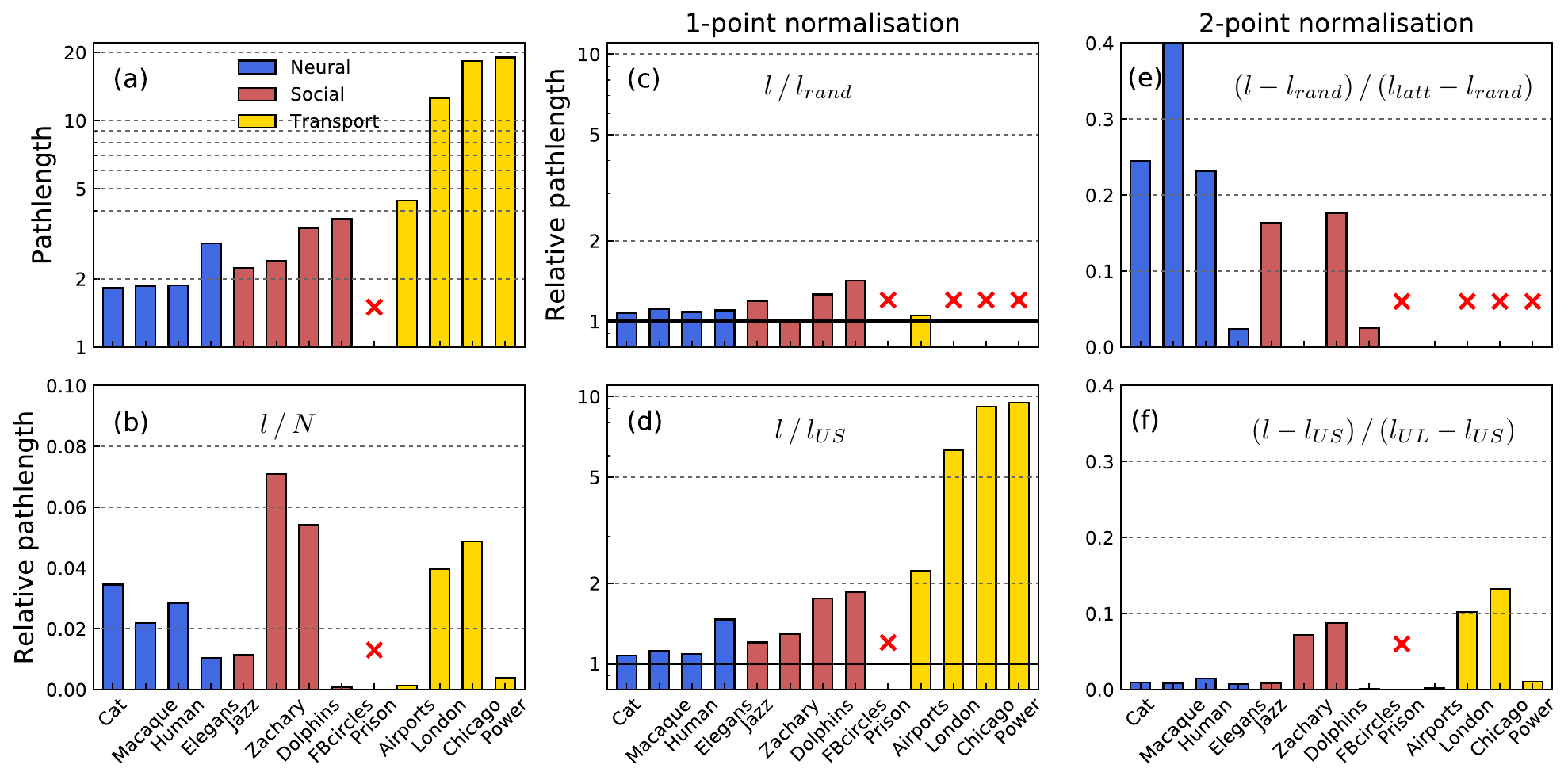}
	\caption{ \label{fig:Pathlen_RealNets}
	{\bf Comparison of absolute and relative pathlengths for selected neural, social and transportation networks.}
	(a) Absolute average pathlength of the empirical networks, (b)-(f) different relative pathlength definitions. (b) Relative to network size $N$, (c) relative to the pathlength to the ultra-short boundary, (d) relative to equivalent random graphs. (e) and (f) 2-point normalisations considering the absolute ultra-short and ultra-long boundaries (e), and relative to  random graphs and ring lattices as benchmark graphs (f). Red crosses indicated cases for which all random graphs generated as benchmark were disconnected and had thus an infinite pathlength. The Prison social network is weakly connected and can thus only be studied by characterising efficiency.
	} 
\end{figure*}

\begin{figure*} 
	\centering
	\includegraphics[width=0.8\textwidth,clip=]{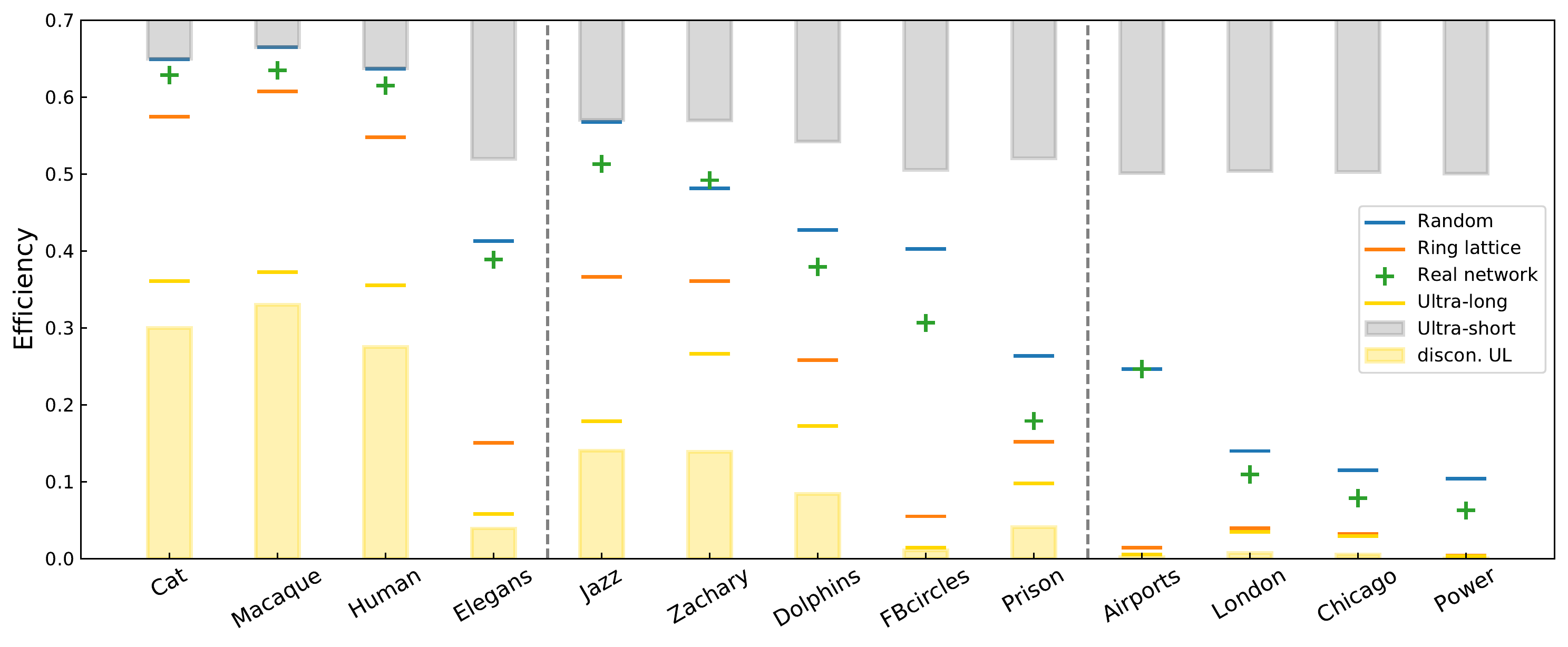}
	\caption{ \label{fig:Effic_RealNets}
	{\bf Comparison of efficiency for selected neural, social and transportation networks.}
	The efficiency of the thirteen empirical networks (+) is shown together with their ultra-short and ultra-long boundaries. The span of the boundaries very much differs from case to case because of the different sizes and densities of the networks studied. For the denser networks (e.g., cortical connetomes) the efficiency of random graphs (blue lines) lie almost on top of the largest possible efficiency (ultra-short boundary). On the contrary, for the sparser networks (e.g., the transportation systems) the efficiency of random graphs very much divert from ultra-short.
	} 
\end{figure*}

We now interpret the results in terms of 1-point and 2-point normalisations. When considering random graphs as the null-hypothesis, $l' = l / l_r$, we find that all empirical networks take values close to $l' \approx 1$, panel (c); with the neural networks, the Zachary karate club and the airports network being the ``shortest'' ones, while the networks of Jazz musicians, the dolphins' social network and the Facebook circles are the ``longest''. The comparison was not possible for three transportation networks (London and Chicago local transportation, and the U.S. power grid) because their densities lie below the percolation threshold and thus no connected random graphs could be constructed of same $N$ and $L$. With these results at hand, we would tend to interpret that all these empirical networks are small-world. However, if contrasted to the actual ultra-short boundary, Eq.~(\ref{eq:1point}), a different scenario is found, panel (d). The lengths of cortical networks (cat, macaque and human) lie marginally above the ultra-short limit. The dolphins and the facebook circle social networks are almost twice as long as the lower boundary and the transportation networks diverge even further, with the London, Chicago and the U.S.A. power grid being more than five times longer than the lower limit.

Taking the 2-point normalisations into account, if random graphs and ring-lattices are considered as the benchmarks, panel (e), the brain connectomes, the collaboration network of jazz musicians and the dolphin's network appear ranked as the longest networks while Zachary Karate Club and the airports network seem to be the shortest. But when normalised according to the ultra-short and the ultra-long boundaries, Eq.~(\ref{eq:2point}), it becomes evident that all the networks are closer to the ultra-short boundary than to the ultra-long, Fig.~\ref{fig:Pathlen_RealNets}(e). The Zachary Karate Club and the dolphins' are the longest social networks while the London and Chicago local transportation networks fall above $10\%$ of the whole range, between the ultra-short and the ultra-long limits. 

The differences displayed between the two choices for 1-point and 2-point normalisations are to be understood in terms of the results shown in Figs~\ref{fig:NetModels}(a) and (b). 
When considering random graphs as the benchmark to compare two empirical networks, we are employing as reference two sets of random graphs (of distinct size and density) whose position with respect to the boundaries may very much differ. For example, the length of one ensemble may depart from the ultra-short limit (if sparse) while the second set of random graphs may lie at the ultra-short limit, if dense enough.

To clarify this further, Figure~\ref{fig:Effic_RealNets} shows the efficiency of the thirteen empirical networks (+), 
together with their corresponding ultra-short (gray bars) and ultra-long (gold bars) boundaries, and the efficiencies of equivalent random graphs (blue lines) and ring lattices (red lines). The span from the upper to the lower limits differs from case to case due to the particular size and density of each network. In the case of the three brain connectomes (cat, macaque and human) their equivalent random graphs match the ultra-short boundary. Thus comparing these networks to random graphs is the same as comparing them to the lower limit. However, for sparser networks this is no longer the case. For example, the efficiency of the neural network of the C. elegans is close to that of equivalent random graphs, but both values depart from the ultra-short boundary. In this case, the network is still far from ring lattices (red lines) and the ultra-long boundary. The opposite scenario is found for the transportation networks. Their efficiency, and the efficiency of their corresponding random graphs, both lie closer to the ultra-long boundary than to the ultra-short. These results elucidate the observations in Figs.~\ref{fig:Pathlen_RealNets}(c) - (f). Although the length of empirical networks is usually comparable to random graphs, the position these values values take with respect to the limits very much differs from case to case, depending on the size and the density of each network.

\section*{Summary and Discussion}			\label{sec:Summary}

Among the many descriptors to characterise complex networks, the average pathlength is probably the most relevant one. It lies at the heart of the small-world phenomenon and also plays a crucial role in network dynamics, as short pathlengths facilitate global synchrony~\cite{Arenas_Review_2008, Boccaletti_Review_2006} or the diffusion of information and diseases~\cite{Pei_SpreadingReview_2013, Pastor_EpidemicReview_2015}. Unfortunately, the pathlength is also difficult to treat mathematically and most analytic results so far are restricted to statistical approximations on scale-free and random graphs~\cite{Chung_Pathlength_2004, Fronczak_AvPathlen_2004}. Here, we have taken a significant step forward by identifying and formally calculating the upper and the lower boundaries for the average pathlength and efficiency of complex networks for all sizes and densities. We provide results for both directed and undirected networks, whether they are sparse (disconnected) or dense (connected), thus delivering solutions that are useful for the whole range of real networks studied in practice beyond singular study cases, e.g., the thermodynamic limit.

We have found that these boundaries are given by specific architectures which we generically refer to as ultra-short (US) and ultra-long (UL) networks. The optimal configurations are not always unique and may vary according to size or density. Ultra-short and ultra-long networks are thus characterised by a collection of models as summarised in Figure~\ref{fig:USULnets}. 
From a practical point of view, our theoretical findings solve the crucial problem of assessing, comparing and interpreting how short (or how long) a complex network is. 
Evaluating the length of a network with a single number -- whether absolute or relative -- has strong limitations and often involves making arbitrary choices. A more telling approach is to display networks together with their boundaries. For example, Figure~\ref{fig:Effic_RealNets} offers a synoptic way to assess the position of the network in the space of efficiencies and thereby discloses all the relations with absolute bounds and usual models. This framework allows for a complete and accurate description and interpretation of the efficiency of complex networks. It can then be supplemented with specific quantities such as the relative measures depicted on Fig.~\ref{fig:Pathlen_RealNets}. We advocate for the representation in Fig.~\ref{fig:Effic_RealNets} whenever a claim about the length of networks is made. 

Future efforts shall be carried to identify the limits of other graph measures and thus contribute to a more reliable framework for the analysis of complex networks. For example, an analysis of the clustering coefficient of the extremal configurations could shed a brighter light on the phenomenon of small-worldness.

For illustration, we have here studied empirical networks from three scientific domains -- neural, social and transportation. The comparison evidences that cortical connectomes are the shortest of the three classes. In fact, they are practically as short as they could possibly be and any alteration of their structure, e.g., a selective rewiring of their links, would only lead to negligible decrease of their pathlength. On the other extreme, transportation networks are more than five times longer than the corresponding lower limit. This contrast between cortical and transportation networks is rather intriguing since both are spatially-embedded. Over the last decade it has been discovered that brain and neural connectomes are organised into modular architectures with the cross-modular paths centralised through a rich-club~\cite{Zamora_Thesis, Zamora_Hubs_2010, Heuvel_HubsHuman_2011, Zamora_FrontReview_2011, Sporns_AttributesSI_2013}. Recently, it has also been shown that this type of organisation supports complex network dynamics as compared to the capabilities of other hierarchical architectures~\cite{Senden_RichClub_2014, Zamora_FComplexity_2016}. Now, we also find that cortical connectomes are quasi-optimal in terms of pathlength. While the aim of neural networks might be the rapid and efficient access to information within the network, transportation networks are developed to service vast areas surrounding a city. Thus they are often characterised by long chains spreading out radially from a rather compact centre. Although transportation networks could never meet optimal average pathlengths for this reason, our results may inspire strategies for their optimisation.

Our results have implications beyond the structural analysis of complex networks. It remains an open question to investigate how dynamic phenomena, e.g., synchrony and diffusion behave in the families of ultra-short and ultra-long networks we have discovered, and to assess their use as benchmarks for the study of network dynamics.

\section*{Methods}			\label{sec:Methods}

\subsection*{Length boundaries of complex networks}		\label{sec:Boundaries}

\noindent {\bf Undirected graphs.} Our first goal is to generate ultra-short (US) graphs, that is, graphs of arbitrary number of nodes and edges with the shortest possible pathlength. This can be achieved by adding edges to a star graph. Indeed, any arbitrary order followed to add edges to a star graph will result in an ultra-short graph. Figure~\ref{fig:USULnets}(a) illustrates two different examples. One consists of seeding edges at random while in the other links are orderly planted favouring the creation of new hubs; a procedure that would eventually lead to the formation of a rich-club. The reason why the order in which edges are added is irrelevant for the value the average pathlength takes, is that the diameter of the star graph is $\delta = 2$. Any further edge $(i,j)$ added results in converting an entry $d_{ij} = 2$ in the distance matrix to $d_{ij} = 1$. As a consequence, at fixed density, all graphs with diameter $\delta \!=\! 2$ are ultra-short and have the same pathlength. See the ultra-short network theorem in Supplementary Text for a formal statement, and Refs.~\onlinecite{Barmpoutis_Extremal_2010, Barmpoutis_Extremal_2011, Gulyas_Pathlength_2011} for alternative proofs. The pathlength and efficiency of a US graph are given by:
	\begin{eqnarray}
	l_{US} & = & 2 - \frac{L}{L_o} \, = \, 2 - \rho	 		\label{eq:PathlenUS} \\
	E_{US} & = & \frac{1}{2} \left[ \, 1 + \frac{L}{L_o} \, \right] \, = \, \frac{1}{2} \big[ \, 1 + \rho \, \big],	\label{eq:EfficUS}
	\end{eqnarray}
where $L_o = \frac{1}{2} N(N-1)$ and $\rho := L / L_o$ is the density of the network.

To generate connected graphs of arbitrary $L$ with the longest possible pathlength, namely ultra-long (UL) graphs, we consider the path graph as a starting point. Any link added to a path graph reduces its diameter, i.e., the distance between the nodes at the two ends. The key is thus to add new links, one-by-one, such that the diameter of the resulting network is minimally reduced at every step. This can be achieved by orderly accumulating all new edges at one end of the chain, Fig~\ref{fig:USULnets}(b). The procedure creates complete subgraphs of size $N_c$ as $L$ grows, with $N_c = \floor[\bigg]{ \frac{1}{2} \left[  3 + \sqrt{9 + 8(L-N)}  \right] }$ where $\lfloor \cdot \rfloor$ stands for the \emph{floor} function. The remainder of the network consists of a tail of size $N_t = N - N_c$. The complete subgraph contains $L_c = \frac{1}{2}N_c(N_c-1)$ edges and the tail $L_t = N_t$. If $L \neq L_c + L_t$, the remaining edges are placed connecting the first node of the tail to the complete subgraph.  We find that the average pathlength of an UL graph can be approximated as:
	\begin{equation}
	l_{UL} \approx 2 + \rho - 2\sqrt{\rho} + \frac{N}{3} \big( 1 - 3 \rho + 2 \rho \sqrt{\rho} \, \big).
	\end{equation} 
The approximation improves as $N$ increases, incurring a relative error smaller than $1\%$ for $N > 122$. See Supplementary Text for the exact solutions (Theorem~3) and Ref.~\onlinecite{Barmpoutis_Extremal_2011}.

So far, we have only considered connected networks. When $L < N-1$ the shortest architecture (largest possible efficiency) consists of an incomplete star graph of size $N' \!=\! L+1$. This leaves the remaining $N-N'$ nodes isolated, Fig.~\ref{fig:USULnets}(a). We refer to these networks as disconnected ultra-short (dUS) graphs. Once $L \geq N$, the solutions for the most efficient and ultra-short graphs are identical (i.e., star graphs with added links).

The construction of disconnected graphs with smallest efficiency is a non-Markovian process. Smallest efficiency is achieved by never having a pair of nodes indirectly connected. This can be realised by forming complete subgraphs which are mutually disconnected. In the special cases when $L = \frac{1}{2} M(M-1)$ for $M = 2, 3, \ldots, N$, the network with smallest efficiency consists of a complete subgraph of size $M$, and $N-M$ isolated nodes, Fig.~\ref{fig:USULnets}(c). The distance between two nodes in the complete subgraph is $d_{ij}=1$ while all other distances are infinite. Therefore, the efficiency in these cases is exactly $\rho=\frac{L}{L_o}$. The efficiency can also be equal to $\rho$ in intermediate cases, see Supplementary Text.  We refer to these networks as disconnected ultra-long (dUL) graphs. In summary, the efficiency of dUS and of dUL graphs are given by:
	\begin{eqnarray} 	
	E_{dUS} &=& 	\frac{1}{4\,L_o} \left[\, L^2 + 3L \, \right] :  \quad L < (N-1),		\label{eq:Effic_dUSgraph}	\\
	E_{dUL} &=& 	L / L_o = \rho :  \quad 0 < L < L_o.						\label{eq:Effic_dUSdigraph}
	\end{eqnarray}

\noindent {\bf Directed graphs:} We will denote the properties of digraphs with a \emph{tilde}, e.g., $\tilde{L}$, $\tilde{l}$ and $\tilde{E}$. Following standard notation, we will refer to directed links as \emph{arcs}. The identification of ultra-short and ultra-long digraphs is more intricate because the conditions for a digraph to be connected are more flexible, distinguishing between weakly and strongly connected. We have found three major differences with the results for graphs. ($i$) The minimally connected digraph is a directed ring (DR) instead of star or path graphs. Thus, directed rings are the origin for both ultra-short and ultra-long connected digraph families. ($ii$) The construction of US and UL digraphs is often a non-Markovian process. ($iii$) In certain regimes of density more than one configuration compete for the optimal pathlength or efficiency.

The ultra-short graph theorem guarantees that any graph with diameter $\delta = 2$ has the shortest possible pathlength regardless of its precise configuration. This result also applies to digraphs and thus any set of arcs added to a star graph will lead to an ultra-short digraph. The difference is that a star graph contains $\tilde{L} = 2\,(N-1)$ arcs. Hence, the result holds for $\tilde{L} \geq 2(N-1)$. However, in the range $N \leq \tilde{L} < 2\,(N-1)$ strongly connected digraphs exist, whose diameter is always larger than two. In this range the digraphs with the shortest pathlength consist of a set of directed cycles overlapping at a single hub, Fig.~\ref{fig:USULnets}(d). We name these networks as \emph{flower digraphs}. Notice that flower digraphs represent the natural transition between a directed ring and a star graph. The DR is the flower made of a unique cycle of length $N$ and a star graph is the flower digraph with $N-1$ ``petals'' of length $2$. Hence, in this regime ultra-short digraph generation is non-Markovian.

Construction of ultra-long digraphs turns rather intricate and we will provide a partial solution here. Numerical exploration with small networks revealed that, in general, more than one optimal configuration exist. See a summary of all the ultra-long digraphs for networks of $N=5$ in Figs.~S9 and~S10. The process is divided into two regimes, with a transition happening at $\tilde{L'} = \frac{1}{2} N(N+1) - 1$, or $\tilde{\rho}' = \frac{1}{2} + \frac{1}{N}$. Given a DR with each node $i$ pointing to node $i+1$ (except the last points to the first) any arc $i \to j$ added in the forward orientation ($i < j$) becomes a shortcut notably reducing the distance between several nodes. Arcs running in the opposite orientation ($i \to j$ with $i > j$) introduce cycles of length $j-i+1$ which only reduce the distance between the nodes participating in the cycle. Thus the strategy is to add arcs to a DR such that each new arc causes the shortest cycle(s) possible. Despite the intricacy of the problem, a particular subclass of digraphs could be found which are guaranteed to be ultra-short. Given an integer $M$, the optimal configuration with $\tilde{L} = N + \frac{1}{2} M(M-1)$ arcs consist of the superposition of a DR and what we name an $M$-backwards subgraph or $M$-BS. An $M$-BS is formed by the first $M$ nodes of the ring, with each node pointing to all its predecessors, Figure~\ref{fig:USULnets}(e). Each $M$-BS contributes to reduce the pathlength of a DR by exactly $\Delta l_{M} = - \frac{1}{\tilde{L}_{0}} \frac{M\left( M-1 \right)}{2} \left[ N-\frac{M+4}{3} \right]$. After calculating the exact solution for these particular cases, we find that the pathlength $l_{UL}$ of ultra-long digraphs, of arbitrary $\tilde{L}$, can be approximated by:
	\begin{equation}
	\tilde{l}_{UL}  \approx  1  +  \frac{3 \tilde{\rho}}{2}  -  \left( \frac{\tilde{\rho}}{3}+1 \right) \sqrt{\frac{\tilde{\rho}}{2}}  +  N \left[ \frac{1}{2} - \tilde{\rho} + \frac{\tilde{\rho}\sqrt{2\tilde{\rho}}}{3} 						\right],	\label{eq:Pathlen_ULdigraph1_approx}
	\end{equation}
This approximation is valid when $\tilde{\rho} < \frac{1}{2} + \frac{1}{N}$. 

In the particular case when $M = N$ ($\tilde{\rho} = \frac{1}{2} + \frac{1}{N}$) the first node receives inputs from all other nodes and the last sends outputs to all the network. All the arcs of the original DR have become bidirectional except for the one pointing from the last to the first node, Fig.~\ref{fig:USULnets}(e). Its pathlength is $\tilde{l}_{UL} = \frac{N+4}{6}$. From this point, any further arc added wiil create a reciprocal link. Then, the longest pathlength is maintained if the arcs of the $M$-backwards subgraphs are symmetrised in the same order they were created. In the specific cases where $\tilde{L}=\frac{N(N+1)}{2}-1+\frac{K(K-1)}{2}$, it is possible to completely bilateralise an $M$-BS with a $K$-forward subgraph of $K$-FS giving:
\begin{equation}
\tilde{l}_{UL} = \frac{4+N}{6} - \frac{K(K-1)}{2 \tilde{L}_o} \left[ N-\frac{2(K+1)}{3} \right]
\end{equation}

Finally, we focus on the efficiency of networks which may be disconnected. Regarding the ultra-short boundary up to three different network configurations compete for the largest efficiency when $\tilde{L} < 2(N-1)$, Figs.~\ref{fig:USULnets}(f) and (g). One of the routes is non-Markovian. It consists of first creating directed rings of growing size until $\tilde{L} = N$ which then naturally continues into flower digraphs. The second route is Markovian and corresponds to the directed version of the disconnected star procedure introduced for graphs. Both routes converge at $\tilde{L} = 2(N-1)$ where a star graph is formed. Figure~\ref{fig:USULnets}(g) shows the competition of the three models for largest efficiency for different network sizes. At larger densities, when $\tilde{L} \geq 2(N-1)$, the ultra-short theorem applies.

To construct digraphs with minimal efficiency, we seed arcs to an initially empty network such that it contains as many weakly connected nodes as possible. We do so by adding $M$-backward subgraphs of increasing $M$ to the empty graph, Figure~\ref{fig:USULnets}(h). The distance matrix of such a digraph contains $\tilde{L}$ entries with $d_{ij}=1$ and all remaining entries are infinite. Thus, its efficiency is $\tilde{E}_{dUL} = \tilde{L} / \tilde{L}_{0} = \tilde{\rho}$. Arcs can be seeded following this procedure until $\tilde{L} = \tilde{L}_o / 2$, corresponding to the largest $M$-BS, with $M = N$. At this point, the network consist of the densest possible directed acyclic graph. Any subsequent arc added will  introduce at least one cycle. To conserve the lowest efficiency possible, new arcs need to cause cycles with a minimal impact over the path. This is achieved, again, by bilateralising the $M$-backwards subgraphs in the forward direction. In these special cases, the efficiency of the digraphs equals their link density: $\tilde{E}_{dUL} = \tilde{\rho}$. Intermediate values of $\tilde{L}$ which do not meet these criteria, may display small departures from $\tilde{E}_{dUL} = \tilde{\rho}$, with the error decreasing as $N$ grows.

\subsection*{Datasets}				\label{sec:Datasets}

Random graphs were generated following the random generator usually known as the $G(n,M)$ model, which guarantees all realisations have the same number of links. In our nomenclature $n \to N$ and $M \to L$ (or $M \to \tilde{L}$). Scale-free networks were generated following the method in Ref.~\onlinecite{Goh_LoadDistribution_2001}. A power exponent of $\gamma =  3.0$ was used. The resulting SF digraphs would display correlated in- and out-degrees but not necessarily identical. The range of densities for scale-free networks was restricted to $\rho \in [0.0001, 0.1]$ because for $\rho > 0.1$ the power-law scaling of the degree distribution is lost due to saturation of the hubs. For each value of density an ensemble of $1000$ realisations was generated. 
All synthetic networks were generated using the package \texttt{GAlib: a library for graph analysis in Python} (https://github.com/gorkazl/pyGAlib).

The empirical networks employed are well known in the literature and have been often used as benchmarks, except for the local transportation of Chicago which we have assembled for the present manuscript. These datasets represent a heterogeneous sample of networks with a variety of sizes and densities, both directed and undirected, see summary in Table~\ref{tab:RealNets}. Those datasets are available online from different sources. We have constructed the  local transportation network of Chicago for the present manuscript by combining the Chicago Transit Authority (CTA) and the METRA commuter rail systems based on the official transportation maps (http://www.transitchicago.com/), Figure~S1. The network consists of 376 stations, of them 142 are serviced by the CTA system and 236 by the METRA railroad. We considered two station to be linked also if they were marked as accessible at a short walking distance, giving rise to a total 402 links and a density of $\rho = 0.006$. Since several stations in the network are named the same, an identifier to the line they belong was added.

\begin{table}
\centering
\begin{tabular}{p{1.5cm} p{2cm} p{1.2cm} p{1.2cm} p{1.2cm} }
\toprule
{\bf Class}	& {\bf Network} 		& {\bf N} 	& {\bf L} 	& {\bf Density}   \\
\hline
Neural	& Cat~\cite{Scannell1993,Scannell1995} 						& 53 		& 826$^*$ 	& 0.300  \\
		& Macaque~\cite{Kotter_Retrieval_2004, Kaiser_Placement_2006}	& 85 		& 2356$^*$ 	& \bf{0.330}  \\
		& Human~\cite{Hagmann_Core_2008}						& 66 		& 590 		& 0.275  \\
		& \emph{C. elegans}~\cite{Varshney_Elegans_2011}				& 274 	& 2956$^*$ 	& 0.040 \\
\hline
Social	& Jazz~\cite{Glaiser_Jazz_2003}					 		& 279 	& 2742	& 0.141 \\
		& Zachary~\cite{Zachary_Karate_1977}						& \bf{34}	& 78		& 0.139 \\
		& Dolphins~\cite{Lusseau_Dolphins_2003}					& 62		& 159	& 0.084 \\
		& FB circles  											& 4039	& 88234	& 0.011 \\
		& Prison~\cite{MacRae_SocioData_1960}						& 67		& 182$^*$	& 0.041 \\
\hline
Tranport	& London~\cite{DeDomenico_London_2013}					& 317 	& 370	& 0.007  \\
		& Chicago												& 376	& 402 	& 0.006  \\
		& Airports~\cite{Guimera_Airports_2005}						& 3618	& 14142	& 0.002  \\
		& Power grid~\cite{Watts_WSmodel_1998}					& \bf{4941}	& 6594 	& \bf{0.0006}  \\
\bottomrule
\end{tabular}

\caption{ 	\label{tab:RealNets}
	{\bf Main characteristics of the sample real networks investigated:} For illustrative purposes we analyse twelve networks from three different domains: neural, social and transportation. Networks marked with and asterisk ($^*$) are directed. The remaining networks are all undirected.
	}
\end{table}

\acknowledgements{
This work was funded by the European Union's Horizon 2020 research and innovation programme under grant agreement No. 720270 (HBP SGA1) and No. 785907 (HBP SGA2).


\clearpage
\newpage
\widetext


\setcounter{equation}{0}
\setcounter{figure}{0}
\setcounter{table}{0}
\setcounter{page}{1}
\setcounter{section}{0}
\makeatletter

\renewcommand{\theequation}{S\arabic{equation}}
\renewcommand{\thefigure}{S\arabic{figure}}
\renewcommand{\bibnumfmt}[1]{[S#1]}
\renewcommand{\citenumfont}[1]{S#1}

\newtheorem{defin}{Definition}
\newtheorem{rem}{Remark}
\newtheorem{lem}{Lemma}
\newtheorem{prop}{Proposition}
\newtheorem{theor}{Theorem}
\newtheorem{cor}{Corollary}

\begin{center}
\textbf{\Large \emph{Supplementary Text for:} \\ ``Sizing the length of complex networks'' \emph{by} }

\vspace{0.5cm}
Gorka Zamora-L\'opez \& Romain Brasselet

\end{center}

\vspace{1cm}

\section{Materials and Methods}		\label{sec:Materials}

\subsection{Synthetic network models}
Synthetic networks were generated using the package \texttt{GAlib: a library for graph analysis in Python} (https://github.com/gorkazl/pyGAlib). Random graphs were generated following the random generator usually known as the $G(n,M)$ model, which guarantees all realisations have the same number of links. In our nomenclature $n \to N$ and $M \to L$ (or $M \to \tilde{L}$). Therefore, we used the function \emph{RandomGraph(N,L)}. Scale-free networks were generated using function \emph{ScaleFreeGraph(N,L,gamma)} which follows the method in Ref.~\onlinecite{Goh_LoadDistribution_2001} and guarantees the right number of edges. A power exponent of $\gamma =  3.0$ was used. In the case of directed scale-free networks, the probability of choosing a vertex either as a target or as a source for an arc followed the same scaling. The resulting SF digraphs would display correlated in- and out-degrees but not necessarily identical. 
For the results in Figure~2, random graphs of density ranging from $\rho = 0.0001$ to $1.0$ were produced. The range of densities for scale-free networks was restricted to $\rho \in [0.0001, 0.1]$ because for $\rho > 0.1$ the power-law scaling of the degree distribution is lost due to saturation of the hubs. For each value of the density an ensemble of $1000$ realisations was generated and the ensemble averaged pathlength and efficiency were calculated.

\subsection{Empirical datasets}

All empirical networks employed are well known in the literature and have been often used as benchmarks, except for the local transportation of Chicago which we have assembled for the present manuscript. These datasets represent a heterogeneous sample of networks with a variety of sizes and densities, both directed and undirected. All datasets are available online from different sources.

The nervous system of the nematode {\bf Caenorhabditis elegans} consists of 302 neurones which communicate through gap junctions and chemical synapses. We use the collation performed by Varshney et al. in Ref.~\onlinecite{Varshney_Elegans_2011}; the data can be obtained at http://wormatlas.org/neuronalwiring.html. After organising and cleaning the data we ended with a network of $N = 274$ neurones and $L = 2956$ directed arcs between them. The network combines both gap junctions, which are bidirectional, and chemical synapses, which are directed. The resulting network has a density of $\rho = 0.040$.
The dataset of the cortico-cortical connections in {\bf cats' brain} was created after an extensive collation of literature reporting anatomical tract-tracing experiments~\cite{Scannell1993, Scannell1995, Hilgetag_Clusters_2000}. It consists of a parcellation into $N = 53$ cortical areas of one cerebral hemisphere and $L = 826$ directed fibre connections between the areas, giving rise to a density of $\rho = 0.300$.
The cortico-cortical connections in the {\bf macaque monkey} are based on a parcellation of one cortical hemisphere into $N = 95$ areas and the fibre projections between them~\cite{Kaiser_Placement_2006}. The dataset, which can be downloaded from http://www.biological-networks.org, is a collation of tract-tracing experiments gathered in the CoCoMac database (http://cocomac.org)~\cite{Kotter_Retrieval_2004}. Ignoring all cortical areas that receive no input we ended with a reduced version of $N = 85$ cortical areas, $L = 2356$ directed fibres and a density of $\rho = 0.330$.
The anatomical {\bf human brain connectome} can be estimated using diffusion imaging and tractography. We considered the dataset published in Ref.~\onlinecite{Hagmann_Core_2008}. The network consists of a parcellation of both hemispheres into 66 regions and $L = 590$ tracts between them. 

We have studied four social networks which are well-known and highly reported in the literature: the {\bf Zachary karate club}~\cite{Zachary_Karate_1977}, the social network of a group of {\bf dolphins}~\cite{Lusseau_Dolphins_2003}, the collaboration network of {\bf Jazz musicians}~\cite{Glaiser_Jazz_2003}, a social network of individuals participating in {\bf Facebook circles}, and a friendship network between {\bf prison inmates} collected in the 1950s~\cite{MacRae_SocioData_1960}.

We have studied two well-known transportation networks, the {\bf world-wide air transportation network} consisting of the world airports connected by a direct flight~\cite{Guimera_Airports_2005} and the {\bf power grid of the USA}~\cite{Watts_WSmodel_1998}. Additionally, we investigated two local transportation networks. The {\bf London transportation network} which combines the London Underground and Overground public transportation lines~\cite{DeDomenico_London_2013}. It is composed of $N= 317$ underground and train stations with 370 links, for a density of $\rho = 0.007$. Finally, we have constructed the {\bf local transportation network of Chicago} for the present manuscript by combining the Chicago Transit Authority (CTA) and the METRA commuter rail systems based on the official transportation maps (http://www.transitchicago.com/). The network consists of 376 stations, of them $142$ are serviced by the CTA system and $236$ by the METRA railroad. For the combined network we considered two station to be linked also if they were marked as accessible at a short walking distance, giving rise to a total 402 links and a density of $\rho = 0.006$. Since several stations in the network are named the same, an identifier to the line they belong was added.


\section{Efficiency of empirical sample networks}		\label{sec:Materials}

In Section~II.C of main text (Figure~4) we have studied the average pathlength and the relative pathlengths of neural, social and transportation networks. For completeness we now show in Fig.~\ref{fig:RealNets_Effic} the same results as in the main text but terms of the efficiency of the networks~\cite{Latora_Efficiency_2001}. Notice that the friendship network of prison inmates could not be studied in terms of its pathlength since it is directed and weakly connected and its pathlength is thus infinite. Also, in Figs.~4(d) and (f), the results for three transportation networks could not be provided because of their sparsity. Their density falls below the percolation threshold for random graphs and thus no connected benchmark graphs could be realised to study them. All these cases, however, can be studied in terms of efficiency, as Fig.~\ref{fig:RealNets_Effic} illustrates. The prison social network is now found to be the least efficient among the social networks studied.

The first difference with the results based on the pathlength is that the absolute efficiency of the real networks is very informative, panel (a). Although the efficiency of a network also depends on its size and density, its values are bounded between zero and one. Thus, efficiency is easier to interpret and compare than average pathlength. For example, the efficiency of transportation networks is found to be very small, with three of them taking $E < 0.1$. 
As found for the pathlength, the efficiency of many networks falls close to that of random graphs, panel (c), what might be interpreted as these networks being almost optimally efficient. However, the comparison to the ultra-short boundary (largest efficiency possible for each $N$ and $L$ combination) clarifies that only the three cortical networks are practically optimal. On the other hand, the efficiency of all transportation networks lies far below the true boundary, panel (d), despite the airports network being as efficient as equivalent random graphs. Indeed, their efficiency very much approaches the ultra-long boundary (smallest efficiency) as evidenced by the 2-point relative efficiency taking values above $0.8$, panel (f).

\begin{figure*} [!hb]
	\centering
 	\includegraphics[width=0.8\textwidth]{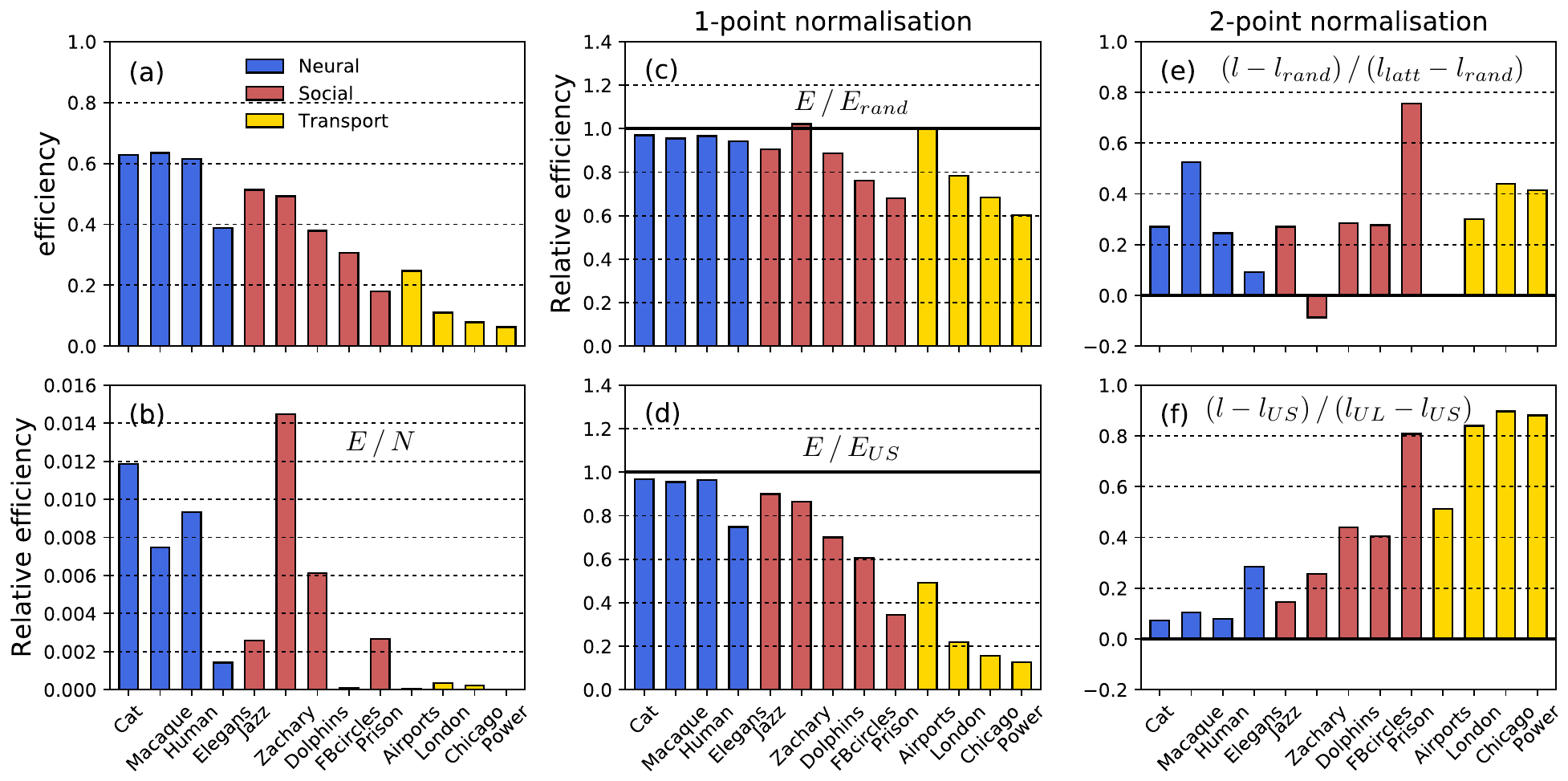}
 	\caption{ \label{fig:RealNets_Effic}
	{\bf Comparison of absolute and relative efficiencies for selected neural, social and transportation networks.}
	(a) Absolute efficiency of the empirical networks, (b)-(f) different relative efficiency definitions. (b) Relative to network size $N$, (c) relative to equivalent random graphs, and (d) relative to the efficiency of ultra-short networks. (e) and (f) 2-point normalisations considering the relative to random graphs and ring lattices as benchmark graphs (e) and the absolute ultra-short and ultra-long boundaries (f).
	} 
\end{figure*}

\clearpage
\section{Boundaries for pathlength and efficiency of graphs}		\label{sec:GraphBoundaries}

We first recall a few basic definitions. An undirected graph $G(N,L)$ is a graph composed of $N$ nodes and $L$ undirected links (edges). A \emph{simple graph} is a graph where nodes are connected by at most one edge. The maximum number of edges a graph can contain is $L_o = \frac{1}{2} N(N-1)$. A \emph{complete graph} is thus the graph with $L_o$ edges and an \emph{empty graph} is a graph with no links ($L = 0$). The \emph{density} $\rho$ of a graph is the fraction of the number of links to the maximum possible, $\rho = \frac{L}{L_o}$. The (geodesic) \emph{distance} $d_{ij}$ between two nodes is the length of the shortest path between them. 
The \emph{distance matrix} $D$ of a graph $G$ is then the $N \times N$ matrix collecting the pairwise distances $d_{ij}$ including the shortest cycles $d_{ii}$ on the diagonal.
The \emph{diameter} of the graph is the distance between the most distant pair of nodes, $\delta(G) = \max_{ij}(d_{ij})$. A \emph{connected graph} is a network in which there is at least one path between every pair of nodes, thus $\delta(G) < N$. A \emph{disconnected graph} is a network in which there is at least one pair of nodes for which there is no path connecting them, and thus $\delta(G) = \infty$. The \emph{average pathlength} $l$ of the graph is the average of the distances $d_{ij}$ ignoring the shorted cycles (diagonal entries of $D$) and the \emph{efficiency} $E$ is the average of the inverse of the distances $\frac{1}{d_{ij}}$. Notice that for graphs the distance matrix is symmetric, $d_{ij} = d_{ji}$, and the diagonal entries $d_{ii}$ are ignored for the calculation of the averages. A digraph $\tilde{G}(N,L)$ is a directed graph of $N$ nodes and $\tilde{L}$ directed links (arcs). The definitions above apply, only that $\tilde{L}_{0} = N(N-1)$ and the distance matrix $D$ is usually asymmetric as the equality $d_{ij} = d_{ji}$ does not necessarily hold.

For convenience in the following proofs, let us first define $n_d$ as the number of pairs of nodes in a graph at distance $d$. That is, the number of entries in the distance matrix for which $d_{ij} = d$. Therefore, the following conservation rule holds:
	\begin{equation}
	\sum_{d=1}^{N-1} n_d = L_o,							\label{eq:ndconserve}
	\end{equation}
The average pathlength and efficiency are calculated as:
	\begin{eqnarray}
	l(G) = \frac{1}{L_o} \, \sum_{d=1}^{N-1} d \times n_d,			\label{eq:ndpathlendef}  \\
	E(G) = \frac{1}{L_o} \, \sum_{d=1}^{\infty} \frac{1}{d} \times n_d.	\label{eq:ndefficdef}
	\end{eqnarray}
These are true for both graphs and digraphs, only that $L_o$ differs in the two cases.

\subsection{Graphs with shortest pathlength}

In the main text we argued that any arbitrary strategy followed to add edges to an initial star graph will result in a graph with the shortest possible average pathlength. To understand why the order of link addition is irrelevant we remind that the diameter of a star graph is $\delta_* = 2$. Any edge $(i,j)$ added to a star graph results in converting one entry of the distance matrix from $d_{ij} = 2$ to $d_{ij} = 1$. As a consequence, all graphs with diameter $\delta \!=\! 2$ have the same average pathlength regardless of their detailed topology. In the following we formalise and demonstrate this result. See Refs.~\onlinecite{Barmpoutis_Extremal_2010, Barmpoutis_Extremal_2011, Gulyas_Pathlength_2011} for alternative proofs.

\begin{theor} [Connected ultra-short graphs]  \label{theor:USgraphs}
Let $G(N,L)$ be a simple and connected graph with $N$ vertices and $L$ undirected edges where $(N-1) \leq L < L_o$. If the diameter of $G$ is $\delta \!=\! 2$, then, the average pathlength $l_{US}$ and efficiency $E_{US}$ of $G$ are:
	\begin{eqnarray}
	l_{US}    &=& 2 - \frac{L}{L_o} = 2 - \rho, 	\label{eq:Pathlen_US} \\
	E_{US}  &=& \frac{1}{2} \left[ 1 + \frac{L}{L_o} \right] = \frac{1}{2} \left[ 1 + \rho 	\right] 	\label{eq:Effic_US}.
	\end{eqnarray}
$l_{US}$ is the shortest average pathlength and $E_{US}$ is the largest efficiency that a connected graph of size $N$ with $L$ edges can have.
\end{theor}

\begin{proof}[Proof of Theorem~\ref{theor:USgraphs}]
Let $G(N,L)$ be a simple and connected graph of $N$ vertices and $L$ edges, with $N-1 < L < L_o$. Assume its diameter is $\delta(G) = 2$. By definition, the distance between any two vertices $i$ and $j$ is $d_{ij} = 1$ if there is an edge $(i,j)$ between them, and $d_{ij} > 1$ otherwise. The number of pairs of vertices at a distance $d=1$ is $n_1 = L$ and because we are assuming that the diameter is $\delta(G) = 2$, all other pairs lie at a distance $d = 2$ of each other. Then, $n_2 = (L_o - L)$ and according to Eq.~(\ref{eq:ndpathlendef}), the average pathlength of $G$ is:
	\begin{equation}
	l(G) = \frac{1}{L_o} \big[ 1\times n_1 + 2 \times n_2 \big] = \frac{n_1 + 2\,n_2}{L_o}. 	\nonumber
	\end{equation}
Substituting $n_1$ and $n_2$ we find that $l(G) = 2 - \frac{L}{L_o}$ where, by definition, $\frac{L}{L_o}$ is the density $\rho$. Substituting again $n_1$ and $n_2$ in Eq.~(\ref{eq:ndefficdef}), we obtain that $E(G) = \frac{1}{2} \left[ 1 + L /L_o \right]$.

As stated above, if the diameter of $G$ is $\delta(G) = 2$, it implies that all elements of the distance matrix take either $d_{ij} = 1$ if there is a link between $i$ and $j$, or $d_{ij} = 2$ if no link exists between the two nodes. Adding an edge to $G$ involves that the corresponding element in distance matrix changes from $d_{ij} = 2$ to $d_{ij} = 1$. The number of entries with $d=1$ increases by one, $n_1(L+1) = n_1(L) + 1$ and the number of entries decreases by one as well, $n_2(L+1) = n_2(L) -1$. Since this same change in $n_1$ and $n_2$ happens for any pair of nodes selected to form the new edge, and since the average pathlength depends only on these numbers, all graphs with $L$ links such that $N-1 < L < L_o$ that feature a star graph will have a pathlength given by Eq.~(\ref{eq:Pathlen_US}).
\end{proof}

\subsection{Graphs with largest efficiency}

Theorem~\ref{theor:USgraphs} shows that the largest efficiency for a connected graph is given by Eq.~(\ref{eq:Effic_US}) but  when $L < N-1$ a graph is necessarily disconnected. We now show that an incomplete star graph, see Fig.~1(a), is the configuration with the largest efficiency. Therefore, we will also refer to incomplete stars as disconnected ultra-short graphs.

\begin{defin} [Incomplete star graph]  			\label{def:IncompleteStar}
Let $N$ and $L$ be an arbitrary size and number of edges satisfying $N \geq 3$ and $1 \leq L < (N-1)$. An incomplete star graph  $G(N,L)$ is a disconnected graph formed by one giant connected component, a star graph of size $N' = L + 1$, and $(N - N')$ isolated vertices. 
\end{defin}

\begin{theor} [Disconnected ultra-short graphs]  \label{theor:dUSgraphs}
Let $G_{dUS}(N,L)$ be an incomplete star graph with $N$ vertices and $L$ edges. The efficiency of $G_{dUS}$ is given by
	\begin{equation}
	E_{dUS} = \frac{1}{4\,L_o} \left[\, L^2 + 3L \, \right], 		\label{eq:Effic_dUS}
	\end{equation} 
and $E_{dUS}$ is the largest efficiency that a graph with $1 \leq L < N-1$ edges can possibly have.
\end{theor}

\begin{proof}[Proof of Theorem~\ref{theor:dUSgraphs}]
Let $G$ be a disconnected ultra-short graph as given in Definition~\ref{def:IncompleteStar}. 
Since the connected part of $G$ is a star graph of size $N_* = L +1$, then in $G$ there are $n_1 = L$ pairs of vertices at distance $d=1$, and $n_2 = L^*_o - L$ pairs at distance $d=2$, where $L^*_o = \frac{1}{2} N_* \, (N_*-1) = \frac{1}{2} L \, (L+1)$. The distance between all other pairs is infinite and thus, they do not contribute to the efficiency. Finally, we have that,
	\begin{equation}
	E(G) = \frac{1}{L_o} \left[  1 \times L + \frac{1}{2} \times (L^*_o - L) \right].	\nonumber
	\end{equation}
Replacing $L^*_o$, we obtain Equation~(\ref{eq:Effic_dUS}).

Now, we demonstrate that $E(G)$ is the upper limit for graphs with $N$ vertices and $L < (N-1)$ edges. We prove it by induction. We start with an incomplete star graph made of a hub $A$ connected to $L$ nodes $\{B_i\}$ and a set of isolated nodes $\{C_i\}$. 
There are four different types of edges that can added to this graph: $(A,C_i)$, $(B_i,B_j)$, $(C_i,C_j)$ and $(B_i,C_j)$. Here are the contributions of each of these edges to the efficiency:
\begin{itemize}
 \item $(A,C_i)$ leads to another incomplete star graph. It changes the efficiency by $\Delta(A,C_i) = \frac{2+L}{2L_o}$.
 \item $(B_i,B_j)$ only changes the distance between these two nodes from $2$ to $1$. Thus $\Delta(B_i,B_j) = \frac{1}{2L_o}$.
 \item $(C_i,C_j)$ changes the distance between these two nodes from $\infty$ to $1$. Thus $\Delta(C_i,C_j) = \frac{1}{L_o}$.
 \item $(B_i,C_j)$ connects all nodes of the incomplete star graph to $C_j$. It is easily computed that $\Delta(B_i,C_j) = \frac{7/2+L}{3L_o}$.
\end{itemize}

Of all these contributions, the one leading to the largest efficiency is the first one, i.e. $(A,C_i)$, for $L \geq 1$. We have thus shown the induction step: if the incomplete star graph is the most efficient graph with $L$ edges, then the incomplete star graph is the most efficient graph with $L+1$ edges. We use the fact that the empty graph is a specific case of an incomplete star as the basis of the induction.
\end{proof}

\subsection{Graphs with longest pathlength}

We now formalise the construction of connected graphs with longest average pathlength (smallest efficiency) and demonstrate their properties. For simplicity, we first define a special case of ultra-long graphs, referred to as ``kite-graphs'' and then we generalise the definition, see Fig.~\ref{fig:ULgraphs}(a).

\begin{defin} [Kite graphs]  			\label{def:Kites}
Let $N$ and $N_c$ be two integers satisfying $N > 2$ and $1 \leq N_c \leq N$. Let $K_{N_c}$ be a complete graph of size $N_c$ and $G_t$ be a path graph of size $N_t = N - N_c$. Then, a $(N,N_c)$-kite is the graph formed by the union of $K_{N_c}$ and $G_t$ via a single extra edge which joins the first vertex of the path graph with one vertex of $K_{N_c}$. By definition, a $(N,N_c)$-kite contains:
	\begin{itemize}
	\item $L_c = \frac{1}{2}N_c\,(N_c-1)$ edges within the complete subgraph, 
	\item $L_t = N_t -1$ edges in the tail (the path subgraph), and 
	\item $L_e = 1$ excess edges which link the two subgraphs.
	\end{itemize}
\end{defin}

\begin{defin} [Ultra-long graphs]  			\label{def:ULgraph}
Let $N$ and $L$ be an arbitrary size and number of edges satisfying $N > 1$ and $(N-1) < L < L_o$. Let $K_{N_c}$ be a complete graph of size $N_c \leq N$, and $G_t$ be a path graph of size $N_t$ with $L_e \geq 1$. An ultra-long graph $G(N,L)$ is the result from merging $K_{N_c}$ and $G_t$ by connecting one end-vertex of $G_t$ to $L_e$ vertices within the $K_{N_c}$ component, where $L_e$ is the number of excess edges. We refer to the $G_t$ component as the `tail' of the ultra-long graph. Given arbitrary $N$ and $L$: 
\begin{enumerate}
\item The size of the complete subgraph $K_{N_c}$ is
	\begin{equation}
	N_c = \floor[\bigg]{ \frac{1}{2} \left[  3 + \sqrt{9 + 8 \, (L-N)}  \right] },	\nonumber
	\end{equation} 
where $\lfloor \cdot \rfloor$ stands for the \emph{floor} function, and it contains $L_c = \frac{1}{2} N_c \, (N_c -1)$ edges.
\item The size of the tail is then $N_t = N - N_c$ and it contains $L_t = N_t - 1$ edges.
\item The number of excess edges is $L_e = L - (L_c +L_t)$. Thus, $L_e \in [1,N_c)$ when $N_c < N$, and $L_e = 0$ if 		and only if $N_c = N$.
\end{enumerate}
\end{defin}

\begin{rem} []  			\label{rem:Types}
Let $G(N,L)$ be an ultra-long graph with $N$ vertices and $L$ edges where $L \in [N-1, L_o)$. Then, $G$ contains vertices of three types: \\
\noindent -- Type $\mathcal{A}$ vertices are those within the complete subgraph which are not directly connected to the first vertex of the tail.  There are $N_c - L_e$ vertices of type $\mathcal{A}$. \\
\noindent -- Type $\mathcal{B}$ vertices are those within the complete subgraph which are connected to the first vertex in the tail. There are $L_e$ vertices of type $\mathcal{B}$. \\
\noindent -- Type $\mathcal{C}$ are the vertices of the tail. There are $N_t$ such vertices, labeled as $c_i$ with $i = 1, 2, \ldots, N_t$, being $c_1$ the only vertex in the tail which is connected to the $L_e$ vertices of type $\mathcal{B}$.

\end{rem}

\begin{rem} []  			\label{rem:Kites}
By definition, we have that:
\begin{enumerate}
\item A $(N,N_c)$-kite is an ultra-long graph $G(N,L)$ with pre-defined $N_c$ and $L_e = 1$. Thus, a kite contains $N_c -1$ vertices of type $\mathcal{A}$, one vertex of type $\mathcal{B}$ and a tail with $N_t = N - N_c$ vertices of type $\mathcal{C}$, label as $c_1, c_2, \ldots c_{N_t}$. 
\item A $(N,N)$-kite is the complete graph $K_N$ of size $N$ with no tail. A $(N,1)$-kite and a $(N,2)$-kite are path graphs of size $N$.
\end{enumerate}
\end{rem}

\begin{figure*} 
	\centering
 	\includegraphics[width=0.8\textwidth,clip=]{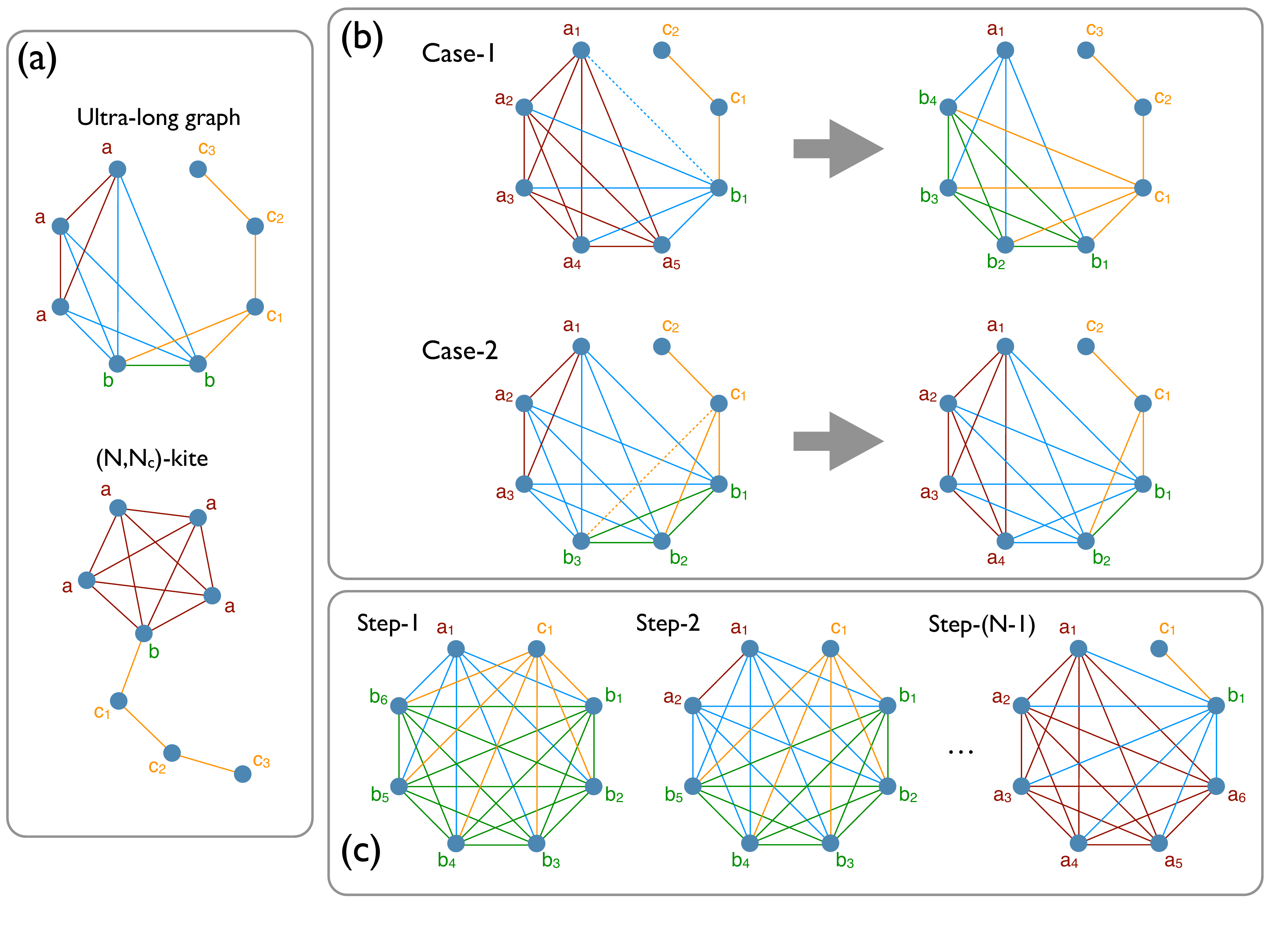}
 	\caption{ 	\label{fig:ULgraphs}
	{\bf Ultra-long graphs.} (a) Illustration of ultra-long and kite graphs. An ultra-long graph is composed of a complete subgraph connected to a path graph by $L_e$ edges. Ultra-long graphs contain three types of vertices ($a$, $b$ and $c$), with an number $L_e$ of type $b$. Connections between $a$-type vertices are labeled in red, edges between $b$-type vertices in green and the connections involving $c$-type vertices in orange. The cross-connections between $a$-type and $b$-type vertices are represented in blue. Kite graphs are special cases of ultra-long graphs in which $L_e = 1$. 
	(b) Edge removal strategies leading to a maximal increase in average pathlength. Case-1 illustrates a $(8,6)$-kite of which a $(a_i,b_1)$ edge is removed, and Case-2 illustrates an ultra-long graph made of a $(8,5)$-kite with two additional edges between the tail and the complete subgraph, of which, one of the $(b_i,c_1)$ edges is removed.
	(c) Illustration of the first $N-1$ steps of the deconstruction iterative procedure from an initial complete graph $K_8$.
	} 
\end{figure*}

Finally, in the following we demonstrate that ultra-long graphs, as defined above are the graphs with largest diameter, longest average pathlength and smallest efficiency that a connected graph of arbitrary $N$ and $L$ can possibly have.

\begin{theor} [Ultra-long graphs]  	\label{theor:ULgraphs}
Let $G(N,L)$ be an ultra-long graph with $N$ vertices and $L$ edges satisfying $N > 1$ and $N-1 \leq L \leq L_o$. Then: \\
	\begin{enumerate}
	\item The diameter of $G$ is 
		\begin{equation}
		\delta_{UL} = N_t + 1 = (N - N_c) + 1 				\label{eq:diam_UL}
		\end{equation}
		and it is the longest diameter that any connected graph with $N$ vertices and $L$ edges 
		can have.
	\item The average pathlength of $G$ is 
	        \begin{equation}
	        l_{UL} = \frac{1}{L_o} \left[ L_c - N_t \, (L-N)  +  \frac{N^3-N_c^3-7N_t}{6}  \right],    \label{eq:Pathlen_UL} 
	        \end{equation}
	and it is the longest average pathlength that any connected graph with $N$ vertices and 
	$L$ edges can have.
    
    	\item The efficiency of $G$ is 
    		\begin{equation}
		E_{UL} = \frac{1}{L_o} \, \bigg[ L-N_t-\frac{L_e-1}{N_t+1} + 
				\; N \, \Big( \psi(N_t+2)+\gamma-1 \Big) \bigg],                 			\label{eq:Effic_UL}
    		\end{equation}
	where $\psi(\cdot)$ is the digamma function and $\gamma \simeq 0.5772$ is the Euler-Mascheroni constant; 
    $E_{UL}$ is the smallest efficiency that any connected graph with $N$ vertices and $L$ edges can have.
	\end{enumerate}

\end{theor}

\begin{proof}[Proof of Theorem~\ref{theor:ULgraphs} ]
We divide the proof of Theorem~\ref{theor:ULgraphs} in two parts. First, we will show that the diameter, average pathlength and efficiency of ultra-long graphs are the expressions given by Eqs.~(\ref{eq:diam_UL}) -- (\ref{eq:Effic_UL}). In the second part we will demonstrate that $\delta_{UL}$ and $l_{UL}$ are the longest diameter and the longest average pathlength a connected graph of $N$ vertices and $L$ edges can have.

The diameter of an ultra-long graph is the length of the path connecting one vertex of type $\mathcal{A}$ to the last vertex of the tail, $c_{N_t}$. Since the vertices of type $\mathcal{A}$ are all equivalent and one step away from the tail, and since there are $N_t$ further steps along the tail to reach $c_{N_t}$, we have that $\delta_{UL} = N_t + 1 = \left( N - N_c \right) + 1$.

To calculate the average pathlength of an ultra-long graph we disentangle how each of the three types of vertices (see Remark~\ref{rem:Types}) contribute to the average pathlength. We find there are four types of contributions. 

Let $D$ be the distance matrix whose elements $d_{ij}$ represent the graph distance between a pair of vertices $i$ and $j$. The distance between any two vertices within the complete subgraph $K_{N_c}$ is $d_{ij} = 1$. This encompasses all distances between nodes of type $\mathcal{A}$ and $\mathcal{B}$. There are $L_c$ such pairs, thus, their contribution to the total sum of lengths is:
	\begin{equation}
	D_1 = L_c \cdot 1 = L_c. 
	\end{equation}

The distance between a vertex of type $\mathcal{A}$ and a vertex of type $\mathcal{C}$, which are labeled as $c_j$: $j = 1, 2, \ldots, N_t$ is $d_{ij} = 1 + j$. The sum of distances from one vertex $i$ of type $\mathcal{A}$ to all vertices in the tail is $\sum_{j = 1}^{N_t} (1 + j) = \frac{1}{2}N_t (N_t+3)$. Since there are $N_c - Le$ vertices of type $\mathcal{A}$, their total contribution is:
	\begin{equation}
	D_2 = (N_c - L_e) \; \frac{1}{2} N_t \, (N_t  + 3).
	\end{equation}

The distance between a vertex of type $\mathcal{B}$ and a vertex of type $\mathcal{C}$, which are labeled as $c_j$: $j = 1, 2, \ldots, N_t$ is $d_{ij} = j$. The sum of distances from one vertex $i$ of type $\mathcal{B}$ to all vertices in the tail is $\sum_{j = 1}^{N_t} j = \frac{1}{2}N_t (N_t+1)$. Since there are $L_e$ vertices of type $\mathcal{B}$, their total contribution is:
	\begin{equation}
	D_3 = L_e \, \frac{1}{2} \, N_t \, (N_t + 1).
	\end{equation}

Finally, the average pathlength between the vertices in the tail is the same as the pathlength of the path graph of size $N_t$. Thus, the contribution to the total pathlength by the tail vertices is $\sum_{n=1}^{N_t - 1} (N_t - n) n $, which simplifying reduces to:
	\begin{equation}
	D_4 = \frac{1}{6} \, \left( N_t -1 \right) \, N_t \, \left( N_t + 1 \right).
	\end{equation}

Having calculated all contributions, the average pathlength of the ultra-long graph is $l_{UL} = \frac{1}{L_o} (D_1 + D_2 + D_3 + D_4)$, which simplifying gives rise to Eq.~(\ref{eq:Pathlen_UL}). The calculation for the efficiency of the ultra-long graphs in Eq.~(\ref{eq:Effic_UL}) follows the same rationale noting that $e_{ij} = 1 / d_{ij}$, and that the digamma function $\psi(n)$ is related to the harmonic numbers $H_n = \sum_{k=1}^n \frac{1}{k}$ as $\psi(n) = H_{n-1} - \gamma$, where $n$ is a positive integer number.  \\

We now prove that the pathlength of ultra-long graphs is the longest pathlength a connected graph with $N$ vertices and $L$ edges can have. We carry out the proof by deconstruction. Starting from a complete graph $K_N$, iteratively at each step ($i$) an edge is removed which maximises the increase in pathlength and ($ii$) we show that the resulting graph is an ultra-long graph. For that we introduce two generic cases of edge removal: \\

\noindent \textit{Case 1}. Consider a $(N,N_c)$-kite with $1 \leq N_c < N$, see example in Figure~\ref{fig:ULgraphs}(b). There are two classes of edges we can remove without disconnecting the graph: ($i$) Edges $(a_i,a_j)$ between any two vertices of type $\mathcal{A}$. The removal of these edges lead to an increase in the pathlength $\frac{1}{L_o}$. And ($ii$) the edges $(a_i,b_1)$ between the only vertex of type $\mathcal{B}$ and the vertices of type $\mathcal{A}$. Their removal leads to an increase in the pathlength of $\frac{N_t + 1}{L_o}$.
The maximal increase in pathlength is thus achieved by removing one of the $(a_i,b_1)$ edges. Since the initial graph is a $(N,N_c)$-kite, the removal of one such edges leads to a large reconfiguration of the vertex types. The initial type $\mathcal{B}$ vertex becomes the new $c_1$ of the tail, which is connected to $(N_c -2)$ vertices in the complete subgraph of size $N_c - 1$ after the edge removal. This leaves a single vertex of type $\mathcal{A}$ converting the rest into type $\mathcal{B}$.  \\

\noindent \textit{Case 2}. Consider a $(N,N_c)$-kite. Let us add $L'$ edges, where $1 \leq L' < (N_c-1)$, between $c_1$ and $L'$ type $\mathcal{A}$ vertices of the complete subgraph, see Figure~\ref{fig:ULgraphs}(b). The result is an ultra-long graph with $L_e = L' + 1$ excess edges. In such a graph, there are two classes of edges which can be removed without disconnecting the graph. ($i$) The edges between any two vertices within the complete subgraph. This includes all edges within and across vertices of type $\mathcal{A}$ and of type $\mathcal{B}$. The removal of any such edge leads to an increase in pathlength of a fraction $\frac{1}{L_o}$. And ($ii$) the edges $(b_i,c_1)$ connecting the complete subgraph with the tail. Their removal leads to an increase in the pathlength of $\frac{N_t}{L_o}$. The maximal increase in pathlength corresponds thus to removing the $(b_i,c_1)$ edges. \\

So far, we have shown that if the $(N,N_c)$-kite is an ultra-long graph, then the $(N,N_c-1)$-kite is as well. And we know the exact graphs in between these cases. Now the proof is finalised by realising that the complete graph is by definition the $(N,N)$-kite. In this particular case, all nodes and all edges are strictly equivalent, therefore, using \emph{Case 1}, we can remove any of the edges between any pair of nodes that we denote $a_1$ and $c_1$. We therefore obtain an ultra-long graph with $N_c=N-1$, $N_t=1$ and $L_e=N-1$.

With this, we have demonstrated that an iterative deconstruction process which leads from a $(N,N_c)$-kite to a $(N_,N_{c-1})$-kite and maximising the increase in average pathlength at each step consists in removing the $N_c-1$ edges touching the tail to get a $(N,N_c-1)$-kite. The resulting graph at each step is also an ultra-long graph as introduced in Definition~\ref{def:ULgraph}. Therefore, adequately alternating \emph{Case 1} and \emph{Case 2}, an optimal deconstruction process exists to transform a complete subgraph [a $(N,N)$-kite], into a path graph [a $(N,2)$-kite] by selectively removing edges in which at each step the gain in average pathlength is maximal. Each step of the process is characterised by an ultra-long graph $G(N,L)$ of $N$ vertices and $L$ edges. \\

The demonstrations that $\delta_{UL}$ and $E_{UL}$ are the longest diameter and the smallest efficiency a connected graph of $N$ vertices and $L$ edges can have, trivially follow from the above demonstration because the diameter is the distance between vertices of type $\mathcal{A}$ and the last vertex in the tail, and because the pairwise efficiency is by definition $e_{ij} = \frac{1}{d_{ij}}$. See Refs.~\onlinecite{Barmpoutis_Extremal_2010, Barmpoutis_Extremal_2011, Gulyas_Pathlength_2011} for alternative proofs.

\end{proof}

\subsection{Graphs with smallest efficiency}		\label{sec:dULgraphs}

Theorem~\ref{theor:ULgraphs} shows that the efficiency of a connected ultra-long graph, Eq.~(\ref{eq:Effic_UL}), is the smallest efficiency a connected graph may have. However, if a graph is disconnected, even for the same $N$ and $L$, a smaller efficiency can be achieved. We have found that the generation of such networks is non-Markovian, meaning that an extremal network with $L + 1$ edges cannot always be achieved by adding one edge to an optimal network with $L$ edges. For certain values of $L$ more than one configuration may exist and compete for the smallest efficiency. Indeed, full clarification was only possible numerically after systematic numerical search for all possible disconnected ultra-long (dUL) graphs in networks of small size. See Section~\ref{sec:BruteForce_ULdigraphs} and Figures~\ref{fig:dULgraphs_1}~--~\ref{fig:dULgraphs_3} for an illustration of all configurations for graphs of $N = 8$. Such numerical investigation reveals 
that, as long as the $N$ nodes and $L$ edges can be decomposed into a set of complete subgraphs, which are disconnected from each other, then the efficiency equals the link density and is minimal. Special cases in which the optimal graph is made of a complete graph of size $M$ and ($N-M$) isolated vertices have been highlighted in the Figures. See also Fig.~1(c).

Although such a decomposition of the edges is not possible for all combinations of $N$ and $L$, the solution dominates for the most part of the range of edge densities, see numerical results in Fig~\ref{fig:Effic_dULgraphs}. The efficiency of exceptional cases deviate little from $E  = \rho$ and thus, in practice, for use with the vast majority of empirical networks known, whose density is $\rho < 1/2$, it is safe to assume that the smallest efficiency possible is $E_{dUL} = \rho$. In the following we formalise and prove these results.

\begin{defin} [$M$-complete disconnected graph]  			\label{def:M_dULgraph}
Let $N>1$ be an arbitrary number of nodes and $M$ an integer satisfying $M \leq N$. An $M$-complete disconnected graph is made of a complete subgraph $K_M$ and $(N-M)$ isolated vertices. Such graph contains $L_M = \frac{1}{2} M(M-1)$ edges.
\end{defin}

\begin{prop} [Disconnected ultra-long graphs \#1]  		\label{prop:dULgraphs_1}
Let $G(N,L)$ be an $M$-complete disconnected graph of $N$ nodes and $L = L_M = \frac{1}{2} M(M-1)$ edges. The efficiency of such a graph is equal to its edge density, $E_{dUL} = \rho$, and $E_{dUL}$ is the smallest efficiency a graph with $N$ nodes and $L_M$ edges can possibly have.
\end{prop}

\begin{proof}[Proof of Proposition~\ref{prop:dULgraphs_1}]
Let $G(N,L)$ be an $M$-complete disconnected graph of $N$ nodes. The distance matrix of such a graph contains $n_1 = L$ entries with $d_{ij} = 1$ (upper triangular values only) corresponding to the distance between the nodes in the complete subgraph. Since all other nodes are isolated, the remaining entries in the distance matrix are $d_{ij} = \infty$. The efficiency of the network is thus calculated as:
	\begin{equation}
	E = \frac{1}{L_o} \sum_{d=0}^N \frac{1}{d} \times n_d = \frac{1}{L_o} \left[ \frac{1}{\infty} \times (L_o - L) + 1 \times L  \right] 
		= \frac{L}{L_o} = \rho.
	\end{equation}
Any edge between nodes $i$ and $j$ sets the distance between them to $d(i,j)=1$. The number of entries in a distance matrix with $d_{ij} = 1$ is thus always $n_1 = L$. In the distance matrix of $G$, all remaining entries take the value $d_{ij} = \infty$. Since they do not contribute to the efficiency, $E = \rho$ is the smallest efficiency a graph could possibly have. The solution proposed here hits this lower bound. Any other circumstance causing at least one of the remaining entries in the distance matrix to take a finite value $1 < d_{ij} < N$, would only increase the efficiency. Consider a configuration of the edges such that two nodes (which are not connected by an edge) would be separated by a distance $d_{ij} = x$ such that $1 < x < N$. The efficiency of such graph would be
	\begin{equation}
	E = \frac{1}{L_o} \left[ 1 \times n_1 + \frac{1}{x} \times n_{N-1} \right] 
		= \frac{1}{L_o} \left[ 1 \times L + \frac{1}{x} \times 1 \right]  
		= \frac{1}{L_o} \left[ L + \frac{1}{x} \right]  >  \frac{L}{L_o}.
	\end{equation}
In conclusion, a graph where the distances between all pair of vertices are infinite, except for those directly connected by an edge, has the smallest efficiency possible.

\end{proof}

The previous result has shown that a sufficient and necessary condition for any graph to have the smallest possible efficiency $E = \rho$ is that its distance matrix contains $n_1 = L$ entries with $d_{ij}=1$ and $n_\infty = (L_o - L)$ entries with $d_{ij} = \infty$. 
We notice that this condition is satisfied by any graph made of several complete subgraphs, which are mutually disconnected from each other. Hence, we now generalise the result:
 
\begin{defin} [$K$-set graph]  			\label{def:Ksetgraph}
Let $N$ and $L$ be an arbitrary number of vertices and edges satisfying $N > 1$ and $L < L_o - N$.
If there exists a set of integers $\{M_i\}$ such that $\sum_i M_i = N$ and $L = \sum_i \frac{1}{2} M_i(M_i-1)$, then a 
$K$-set graph is made of the union of such complete subgraphs. The graph contains $m = |\{M_i\}|$ subgraphs, each of size $\{M_i\}$.
\end{defin}

Notice that the $M$-complete disconnected graphs in Definition~\ref{def:M_dULgraph} are a special case of this more general construction when only one $M_i$ is strictly larger than $1$. Notice also that for some pairs $(N,L)$, more than one decomposition of the $L$ edges into complete subgraphs may be possible. See for example the cases for $L = 3, 4, 6$ and $7$ in Fig.~\ref{fig:dULgraphs_1}. We now formalise and proof that such graphs have the lowest efficiency possible.

\begin{prop} [Disconnected ultra-long graphs \#2]  		\label{prop:dULgraphs_2}
Let $G(N,L)$ be a $K$-set graph of $N$ nodes and $L$ edges as given in Definition~\ref{def:Ksetgraph}. The efficiency of such a disconnected ultra-long graph is  $E_{dUL} = \rho$ and $E_{dUL}$ is the smallest efficiency a graph with $N$ nodes and $L$ edges can possibly have.
\end{prop}

\begin{proof}[Proof of Proposition~\ref{prop:dULgraphs_2} ]
The distance matrix of $G(N,L)$ contains $n_1 = L$ entries with $d_{ij} = 1$, corresponding to the links between the nodes in a complete subgraph, and all other entries take $d_{ij} = \infty$. Hence, following the proof of Proposition~\ref{prop:dULgraphs_1}, it is trivial to show that $E = \rho$ is the smallest efficiency a graph can take.
\end{proof}

\vspace{0.5cm}

Based on numerical observations, we stated before that, for most practical applications, it is safe to consider $E_{dUL} = \rho$ when $\rho < 1/2$. We end this section by computing the largest error incurred when making this assumption. Inspection of results in Fig.~\ref{fig:Effic_dULgraphs} indicate that the largest deviation of the empirical results from $E = \rho$ always happens when $L^* = \frac{1}{2} (N-1)(N-2) + 1$. To understand why, we point at the solutions shown in Fig.~\ref{fig:dULgraphs_2} for $N=8$. The solution for $L = 21$ represents the $(N-1)$-complete subgraph in which a single isolated node remains. This  configuration contains $L_{N-1} = \frac{1}{2} (N-1)(N-2)$ edges. Adding one edge to this graph results in a connected graph by linking the last isolated vertex to one of the nodes in the $(N-1)$-complete component, see configuration for $L = 22$ in Fig.~\ref{fig:dULgraphs_2}. Notice that this graph is, indeed, a $(N,N-1)$-kite graph. Its efficiency is
	\begin{equation}
	E^* = \frac{1}{L_o} \left[ L^* + \frac{1}{2} (N-2) \right] 		\nonumber
	\end{equation}
because there are $n_1 = L^*$ entries with $d_{ij} = 1$ and the formerly isolated vertex is now at a distance $d_{ij} = 2$ from $n_2 = N-2$ nodes. This is the absolute worst case in terms of efficiency as, suddenly, $N-2$ distances strictly larger than $1$ appear. If we assumed the efficiency to be given by the density, at this point, it would incur an error $\Delta(L^*) = E^* - \rho = \frac{N-2}{N(N-1)}$.
The relative difference $\Delta(L^*) / \rho$ at $L^*$ decays with network size as $N \to \infty$, and it becomes smaller than $1\%$ for graphs of size $N > 100$. With this, we have shown that the largest possible error made when assuming that the smallest efficiency of a disconnected ultra-long graph equals its density is bounded by a term that quickly decreases with network size.

\clearpage
\section{Boundaries for pathlength and efficiency of directed graphs}		\label{sec:DigraphBoundaries}

We now turn our attention to directed graphs (digraphs). We will denote the properties of digraphs with a \emph{tilde}, e.g., $\tilde{L}$, $\tilde{l}$ and $\tilde{E}$. and we will refer to directed links as \emph{arcs}. We remind that 
the number of possible arcs a digraph can host is $\tilde{L}_{o} = N(N-1)$ and the distance matrix $D$ is usually asymmetric because the equality $d_{ij} = d_{ji}$ does not necessarily hold. We also remind that the sparsest strongly connected digraph is a directed ring (DR), a network formed by $\tilde{L} = N$ arcs, all pointing in the same orientation.

\subsection{Digraphs with shortest pathlength}

Theorem~\ref{theor:USgraphs} states that any graph with diameter $\delta = 2$ has the shortest possible pathlength regardless of its precise topology. This result also applies to digraphs and thus any arc added to a star graph leads to a digraph with the shortest possible pathlength. In terms of digraphs, a star graph is made of $\tilde{L}_* = 2\,(N-1)$ arcs but the sparsest connected digraph is a directed ring with $\tilde{L}_{DR} = N$ links. In the range $N \leq \tilde{L} < 2\,(N-1)$ the diameter of any digraph is larger than two; hence, the ultra-short theorem does not apply and we need to find the optimal solution valid for this regime. We have found that in this case the optimal solution is given by a novel digraph architecture we named as \emph{flower digraphs}, see Fig.~1(d) in main text. In the following, we restate the ultra-short theorem as applied for digraphs. Then we will introduce flower digraphs as the model with shortest pathlength for digraphs with $\tilde{L} \in [N, 2\,(N-1)]$.

\begin{theor} [Connected ultra-short digraphs]  \label{theor:USdigraphs}
Let $\tilde{G}(N,\tilde{L})$ be a simple and connected digraph of $N$ vertices and $\tilde{L}$ directed arcs where $\tilde{L}  \in [2(N\!-\!1), \tilde{L}_o]$. If the diameter of $\tilde{G}$ is $\tilde{\delta} \!=\! 2$, then the average pathlength $\tilde{l}_{US}$ and efficiency $\tilde{E}_{US}$ of $\tilde{G}$ are:
	\begin{eqnarray}
	\tilde{l}(\tilde{G}) &=& 2 - \frac{\tilde{L}}{\tilde{L}_o}   =   2-\rho, \quad  		\label{eq:Pathlen_USdigraph} \\
	\tilde{E}(\tilde{G}) &=& \frac{1}{2} \left[ 1 +  \frac{\tilde{L}}{\tilde{L}_o} \right]  =  \frac{1}{2}(1+\rho) \label{eq:Effic_USdigraph}.
	\end{eqnarray}
$\tilde{l}_{US}$ is the shortest average pathlength and $\tilde{E}_{US}$ is the largest efficiency that a digraph of size $N$ with $\tilde{L}$ arcs can have.
\end{theor}

\begin{proof}[Proof of Theorem~\ref{theor:USdigraphs} ]

The proof is follows the one of Theorem~\ref{theor:USgraphs}, noting that the number of arcs in a star digraph is $2\, (N-1)$. By definition the distance between two nodes $i$ and $j$ is $d(i,j)=1$ if there is an arc running from $i$ to $j$, otherwise $d(i,j)>1$. Since we assumed that the graph has a diameter $\tilde{\delta}=2$, the distances between pairs of nodes can only take values $1$ and $2$. There are exactly $n_1=\tilde{L}$ pairs with a distance of $d=1$ and $n_2=\tilde{L}_o-\tilde{L}$ with a distance of $d=2$.Therefore

	\begin{equation}
	\tilde{l}(\tilde{G}) = 1 \times \tilde{L} + 2 \times (\tilde{L}_o-\tilde{L})	\nonumber
	\end{equation}
that directly leads to Eq.~(\ref{eq:Pathlen_USdigraph}). The proof for the efficiency follows because, by definition, $E$ is the average of the $\frac{1}{d_{ij}}$ values.

\end{proof}

We now fill the gap for connected ultra-short digraphs in the range $\tilde{L} \in [N, 2\,(N-1)]$ by introducing the flower digraph model.

\begin{defin} [Flower digraphs]  			\label{def:FlowerDigraphs}
Let $N$ and $\tilde{L}$ be arbitrary numbers of nodes and arcs satisfying $N>1$ and $N \leq \tilde{L} \leq 2(N\!-\!1)$.
Let $\{\tilde{G}_{DR}\}$ be a set of $m = \tilde{L} - (N-1)$ directed cycles, all following the same orientation. Let $n = \floor{\tilde{L} / m}$ be the size of the shortest cycle and $n' = n+1$. Let $m_{n'} = \tilde{L} - n \cdot m$ the number of cycles ($m_{n'}$ can be null) of length $n'$ and $m_n = m - m_{n'}$ the number of cycles of length $n$.

A flower digraph is the network resulting from the union of the cycles in the set $\{ \tilde{G}_{DR}\}$ where the union consists of all cycles overlapping onto a single vertex, the hub of the flower digraph. The corresponding directed cycles are thus referred as the ``petals'' of the digraph.
\end{defin}

\begin{rem} [Degrees of flower digraphs]  			\label{rem:DegreesFlowerDigraphs}
Let $\tilde{G}(N,\tilde{L})$ be a flower digraph of $N$ vertices and $\tilde{L}$ arcs. Let $m = \tilde{L} - (N-1)$ be the number of petals in $\tilde{G}$. Then, the input and output degrees of the central hub are $k^{-}_{hub} = k^{+}_{hub} = m$, and the degrees of any other vertex are $k^{-} = k^{+} = 1$. 
The reciprocal degree of all nodes is $k^\leftrightarrow = 0$.
\end{rem}

\begin{rem} [Special cases]  			\label{rem:SpecialFlowers}
Let $\tilde{G}(N,\tilde{L})$ be a flower digraph of $N$ vertices and $\tilde{L}$ arcs. Then, 
	\begin{itemize}
	\item[--] A directed cycle is the sparsest flower digraph possible, made of a unique petal ($m = 1$) of length $n = N$.
	\item [--] A star graph is the densest flower digraph possible, made of $m = N\!-\!1$ petals of length $n = 2$.
	\end{itemize}
\end{rem}

\begin{rem} [Average pathlength of flower digraphs]  			\label{rem:Pathlen_FlowerDigraph}
Let $\tilde{G}(N,\tilde{L})$ be a flower digraph of $N$ vertices and $\tilde{L}$ arcs. Then, the distance matrix of a flower digraph is a block matrix, diagonal blocks representing the distances within the nodes of a cycle, and off-diagonal blocks representing the distances between nodes in different cycles. Given that:
	\begin{itemize}
	\item[--] $D(x) = \frac{1}{2} x^2 (x-1)$ is the sum of pair-wise distances within a cycle of arbitrary size $x$, and 
	\item[--] $D(x,y) = \frac{1}{2}(x - 1) (y - 1) (x + y)$ is the sum of pair-wise distances between the nodes in two different cycles of arbitrary lengths $x$ and $y$, which overlap in a single node,
	\end{itemize}
then, the average pathlength of a flower digraph is calculated summing the contributions $D(n)$ of cycles of length $n$, the contributions $D(n')$ of cycles of length $n'=n+1$ and the cross-contributions $D(n,n')$ from pairs of nodes in cycles of length $n$ and $n'$:
	\begin{equation}
		\tilde{l}_{US}  =  \frac{1}{\tilde{L}_o} \left( S_n + S_{n'} + S_{nn'} \right)	\label{eq:Pathlen_FlowerDigraphs}
	\end{equation}
where, 
	\begin{eqnarray}
		S_n 		&=& 	m_n \,  D(n) \, + \, m_n \left(m_n - 1 \right) \, D(n,n),			\label{eq:Sx_Pathlen}	\\ 
		S_{n'} 	&=& 	m_{n'} \, D(n') \, + \, m_{n'} \left(m_{n'} - 1 \right) \, D(n',n'), 	\label{eq:Sy_Pathlen}	\\
		S_{nn'} 	&=& 	2\, m_n \, m_{n'} \, D(n,n').								\label{eq:Sxy_Pathlen}
	\end{eqnarray}
The diameter is the sum of the lengths of the two longest petals minus 2.
\end{rem}

\begin{rem} [Efficiency of Flower Digraphs]  			\label{rem:Effic_FlowerDigraph}
Let $\tilde{G}(N,\tilde{L})$ be a flower digraph of $N$ vertices and $\tilde{L}$ arcs. Then, the distance matrix of a flower digraph is a block matrix, diagonal blocks representing the distances within the nodes of a cycle, and off-diagonal blocks representing the distances between nodes in different cycles. Given that:
	\begin{itemize}
	\item[--] $E(x) = x\, \left[ \psi(x) + \gamma \right]$ is the sum of inverse pair-wise distances within a cycle of arbitrary size $x$, and 
	\item[--] $E(x,y) = (x+y - 1) \, \psi(x+y) - x\psi(x) - y\psi(y) - (\gamma +1)$ is the sum of inverse pair-wise distances between the nodes in two different cycles of arbitrary lengths $x$ and $y$, which overlap in a single node,
	\end{itemize}
where $\psi(\cdot)$ is the digamma function and $\gamma \simeq 0.5772$ is the Euler-Mascheroni constant. Then, the efficiency of a flower digraph is calculated summing the contributions $E(n)$ of cycles of length $n$, the contributions $E(n')$ of cycles of length $n'=n+1$ and the cross-contributions $E(n,n')$ from pairs of nodes in cycles of length $n$ and $n'$:
	\begin{equation}	
		\tilde{E}_{US}  =  \frac{1}{\tilde{L}_{0}} \left( S_n + S_{n'} + S_{nn'} \right) 	\label{eq:Effic_FlowerDigraphs}
	\end{equation}
where, 
	\begin{eqnarray}
		S_n 		&=& 	m_n \,  E(n) \, + \, m_n \left(m_n - 1 \right) \, E(n,n),			\label{eq:Sx_Effic}	\\ 
		S_{n'} 	&=& 	m_{n'} \, E(n') \, + \, m_{n'} \left(m_{n'} - 1 \right) \, E(n',n'), 		\label{eq:Sy_Effic}	\\
		S_{nn'} 	&=& 	2\, m_n \, m_{n'} \, E(n,n').								\label{eq:Sxy_Effic}
	\end{eqnarray}
\end{rem}

\begin{prop} [Connected and sparse ultra-short digraphs]  \label{prop:USdigraphsFlowers}
Let $\tilde{G}(N,\tilde{L})$ be a flower digraph of $N$ vertices and $\tilde{L}$ arcs as in Definition~\ref{def:FlowerDigraphs}. Then,
	\begin{enumerate}
	\item The pathlength $\tilde{l}_{US}$ of a flower digraph is given by Eq.~(\ref{eq:Pathlen_FlowerDigraphs}) and, $\tilde{l}_{US}$ is the shortest pathlength that a connected digraph with $N$ nodes and $\tilde{L}$ arcs can possible have.
	\item The efficiency $\tilde{E}_{US}$ of a flower digraph is given by Eq.~(\ref{eq:Effic_FlowerDigraphs}) and, $\tilde{E}_{US}$ is the largest efficiency that a connected digraph with $N$ nodes and $\tilde{L}$ arcs can possible have.

	\end{enumerate}
\end{prop}

\subsection{Digraphs with largest efficiency}

Theorem~\ref{theor:USdigraphs} and Proposition~\ref{prop:USdigraphsFlowers} show that the largest efficiency for connected digraphs are given by Eqs.~(\ref{eq:Effic_FlowerDigraphs}) --~(\ref{eq:Sxy_Effic}) and Eq.~(\ref{eq:Effic_USdigraph}) for the cases in which $\tilde{L} \in [N, 2(N\!-\!1)]$ and $\tilde{L} \geq 2(N\!-\!1)$ respectively. At low densities, digraphs are usually disconnected and thus they can only be characterised by their efficiency. Unfortunately, we have found that for $\tilde{L} < 2(N\!-\!1)$ three different digraph configurations compete for the largest efficiency, see Fig.~1(f) of main text. One of the competing models is the flower digraphs introduced in Definition~\ref{def:FlowerDigraphs}. The two remaining models consist of partial directed rings and star digraphs, both aiming at maximising the size of the largest connected component in the digraph. Because the problem does not have a closed form and the solution depends on the size and number of arcs in the network, see Fig.~1(g), here we restrict to formally introducing the two remaining models and providing their efficiencies.

\begin{defin} [Incomplete directed ring]  			\label{def:iDR}
Let $N$ and $\tilde{L}$ be arbitrary numbers of nodes and arcs with $\tilde{L} < N\!-\!1$. An incomplete directed ring is made of the union of a directed ring of size $N' = \tilde{L}$ and a set of $(N - N') = (N - \tilde{L})$ isolated vertices. 
\end{defin}

\begin{rem} []  			\label{rem:Effic_iDR}
The efficiency of an incomplete directed ring is given by:
	\begin{equation} 
		\tilde{E} =  \frac{\tilde{L}}{\tilde{L}_{0}} \left[ \psi(\tilde{L}) + \gamma \right].	\label{eq:Effic_dUSdigraph1}
	\end{equation}
where $\psi(\cdot)$ is the digamma function and $\gamma \simeq 0.5772$ is the Euler-Mascheroni constant.
\end{rem}

\begin{defin} [Incomplete star digraph]  			\label{def:iStarDigraph}
Let $N$ and $\tilde{L}$ be arbitrary numbers of nodes and arcs with $\tilde{L} < 2(N\!-\!1)$. The strongly connected part of an incomplete star digraph is formed by a star graph of size $N' = L +1$ where $L = \floor{\tilde{L} / 2}$ is the number of undirected edges. 
	\begin{itemize}
	\item[--] If $\tilde{L}$ is `even', the remaining $N - N'$ vertices are isolated.
	\item[--] If $\tilde{L}$ is `odd', the remaining arc connects the central hub with one of the isolated vertices in any of the two directions. The final digraph thus contains one weakly connected vertex and $N-N'-1$ isolated vertices. 				\end{itemize}
\end{defin}

\begin{rem} []  			\label{rem:Effic_iStarDigraph}
The efficiency of an incomplete star digraph is:
	\begin{equation}	\label{eq:Effic_iStarDigraph}
	\tilde{E} = \frac{1}{\tilde{L}_o} \left[  \tilde{L} + \frac{1}{2} L \, (\tilde{L} - L -1)  \right].
	\end{equation}
Depending on the value $\tilde{L}$ takes, the expression reduces to:
	\begin{eqnarray}
	\tilde{E} &=& \frac{1}{2 \tilde{L}_o} \left[\, L^2 + 3L \, \right] 	  			\quad \textrm{if $\tilde{L}$ is `even',} \\
	\tilde{E} &=& \frac{1}{2 \tilde{L}_o} \left[ L^2 + 4L + 2  \right]   	\quad \textrm{if $\tilde{L}$ is `odd'.}
	\end{eqnarray}
In the parametrisation of digraphs, these expressions can be rewritten as:
	\begin{eqnarray}
	\tilde{E} &=& \frac{1}{8 \tilde{L}_o} \left[\, \tilde{L}^2 + 3\tilde{L} \, \right] 	  			\quad \textrm{if $\tilde{L}$ is `even',} \\
	\tilde{E} &=& \frac{1}{8 \tilde{L}_o} \left[ 8\tilde{L} + (\tilde{L}-1)^2  \right]   	\quad \textrm{if $\tilde{L}$ is `odd'.}
	\end{eqnarray}
\end{rem}

\begin{rem}
In the range when $\tilde{L} \leq N$, the results of the competition between the incomplete directed ring and incomplete star digraph are:
\begin{itemize}
	\item[--] For any value of $\tilde{L} \leq 23$, the incomplete directed ring has larger efficiency,
	\item[--] For any value of $\tilde{L} > 23$, the incomplete star digraph has larger efficiency. 
	\end{itemize}
\end{rem}

\begin{proof}
 This comes trivially by comparing the efficiencies in the two cases.
\end{proof}

\subsection{Digraphs with longest pathlength}

The challenge in this section is to identify the directed graphs with the longest possible pathlength. Given a directed ring, any additional arc $i \to j$ will give rise to another cycle within the ring shortening the distance between several nodes. Thus, the goal is to identify which arc(s) give rise to extra cycles with a minimal impact on the path structure of the network. An exact solution to this problem turns rather intricate. We performed an exhaustive numerical exploration with small digraphs to understand the problem better. See Section~\ref{sec:BruteForce_ULdigraphs}, and Figs.~\ref{fig:AllULdigraphsN5_p1}~--~\ref{fig:AllULdigraphsN5_p2} for all configurations of digraphs with largest pathlength in networks of $N=5$. In general, we see that more than one optimal ultra-long digraph configuration exist for each value of $\tilde{L}$ but patterns for certain values and ranges of $\tilde{L}$ exist. For example, in the cases when $\tilde{L} = N + \frac{1}{2}M(M-1)$, there is a unique optimal configuration which consists of the directed ring and all the extra arcs gathered among the first $M$ nodes, similarly to the configuration leading to ultra-long graphs (see Kite graphs), but with all arcs pointing in the opposite orientation to the ring, see Fig.~1(e) of main text and highlighted configurations in Figs~\ref{fig:AllULdigraphsN5_p1} and~\ref{fig:AllULdigraphsN5_p2}. We will refer to these sets of arcs as $M$-backwards subgraphs or $M$-BS. 

For the intermediate values of $\tilde{L}$, in between consecutive $M$-BS configurations, precise solutions can become rather difficult and we will hence provide an approximation to estimate the largest pathlength for any value of the density.
The $M$-BS construction works until $M = N-1$. At this point the next exact solution is slightly different because the link $N \to 1$ already exists as part of the DR, see Fig 1(e) bottom leftmost. This solution therefore involves $\tilde{L} = N + \frac{1}{2} N(N-1) - 1$ arcs, or in terms of density, $\rho = \frac{1}{2} + \frac{1}{N}$. Notice that at this point all arcs running in the opposite orientation of the ring have already been placed and any subsequent arc $i \to j$ added to this network will necessary follow the orientation of the ring, that is $i < j$. Finally, in the denser regime of connectivity, when $\tilde{L} \geq N + \frac{1}{2} N(N-1) - 1$ generation of ultra-long digraphs becomes easier. The numerical exploration shows that more than one optimal configuration may co-exist but among them we find one that follows a Markovian process. It consists of orderly bilateralising the arcs of the $M$-backwards subgraphs placing arcs in the forward direction, one after another, smoothly transitioning from $M$-BS of consecutive $M$, lower panel of Fig.~1(e).  

In the following, we will formalise all these results starting from the particular solutions involving a $M$-backwards subgraph. Then we will formalise the result for the particular, bordering case when $\tilde{L} = N + \frac{1}{2} N(N-1) - 1$ and we will finally summarise the ultra-long configurations for the densest cases, when $\tilde{L} > N + \frac{1}{2} N(N-1) - 1$. But first of all we will introduce, for convenience in the following proofs, a definition of the average pathlength without normalisation.

\begin{defin} [Total pathlength]  			\label{def:TotalPathlen}
Let $\tilde{G}(N,\tilde{L})$ be a digraph of arbitrary size $N$ and number of arcs $\tilde{L}$. Let $D$ be the pairwise distance matrix of $\tilde{G}$ with entries $d_{ij}$. Then, the total pathlength $P$ of $\tilde{G}$ is:
	\begin{equation}
	P = \sum_{i=1}^N \sum_{j=1, i \neq j}^N d_{ij},				\label{eq:TotalPathlen}
	\end{equation}
and the contribution of each node to the total pathlength is:
	\begin{equation}
	P_i = \sum_{j=1, i \neq j}^N d_{ij},						\label{eq:NodePathlen}
	\end{equation}

\end{defin}

\begin{defin} [$M$-backwards subgraph]  			\label{def:MBS}
Let $\tilde{G}$ be a digraph of size $N$ with labeled vertices and ordering $v_1, v_2, \ldots v_N$. Let $M$ be an integer satisfying $M \leq N$. Then, an $M$-backwards subgraph (or simply $M$-BS) consists of a subset of $M$ consecutive vertices $V_M = \{v_{k}, v_{k+1}, \ldots, v_{k+M} \}$ of $\tilde{G}$, with each vertex  pointing to all its predecessors within the set. The arc-set of an $M$-BS is thus $A_M =\{(v_i,v_j): \,i = k, \ldots, k+M$ \textrm{and} $j < i \}$. An $M$-BS contains $L_M = \frac{1}{2} M(M-1)$ arcs.
Unless otherwise stated, it shall be understood that $k=1$ and $v_k$ is the first node in the ordering of the digraph.
\end{defin}

We first give evidence for the presence of a DR in any of the ultra-long digraph solutions. Let us consider a digraph of size $N$ and $\tilde{L} = N+1$. With these conditions only two strongly connected digraphs can be built: ($i$) a DR with a single bilateralised arc (a $2$-backwards subgraph) or ($ii$) a DR of $N-1$ nodes and a bilateralised branch off of it. The total pathlength of the first graph is
	\begin{equation}
 	P^{(i)} = \frac{1}{2} N^2(N-1) - (N-2),	
	\end{equation}
where the second term is the contribution of the single $2$-BS. The pathlength of the second graph is
	\begin{equation}
	P^{(ii)} = \frac{1}{2} (N-1)^2(N-2) + (N-1)N. 
	\end{equation}
It is easy to show that $P^{(i)} > P^{(ii)}$ except for the degenerate cases $N={2,3}$ where they are equal. In addition $P^{(i)} > P^{(ii)}$ scales as $\frac{N^2}{2}$. Although this argumentation is not a formal proof, our intuition is that creating a bilateralised branch would always decrease the pathlength more than just adding an arc to the DR.

Now we will show why a DR with an $M$-BS added is always an ultra-long digraph. In a directed ring, the sum of distances from one node $v_i$ to all others, Eq.~(\ref{eq:NodePathlen}), is $P_i = \frac{1}{2} N(N-1)$. We will first show that an arc added to a directed ring has minimal impact on the total pathlength if it forms a $2$-BS. Such an arc is the reciprocal to one of the existing arcs in the DR.

\begin{lem}
Let $\tilde{G}$ be a directed ring of size $N$ and vertex ordering $v_1, v_2, \ldots v_N$. The best addition of a single arc to $\tilde{G}$ such that the impact on the total pathlength is minimal, is an arc $(v_i, v_{i-1})$ forming a $2$-BS. Such an arc reduces the total pathlength of $\tilde{G}$ by $\Delta(P) = N-2$.
\end{lem}

\begin{proof}
Without loss of generality, let's consider a link from $A_k$ to $A_1$, where $k$ runs from $2$ to $N-1$. All the paths starting from the nodes $A_{k+1}$, $A_{k+2}$,... $A_1$  remain the same. But the paths starting from nodes $A_2$ to $A_k$ are changed. Node $A_j$ (where $j$ runs from $2$ to $k$) sees $j-1$ paths reduced by $N-k$, for a total decrease in pathlength of $\delta(k)=-\frac{1}{2} (N-k)k(k-1)$. 

We now want to know where this function is minimal on $k$ from $2$ to $N-1$. This is a third-degree equation with positive third-degree coefficient that vanishes on $k*=\{0,1,N\}$. So, in the range $[2:N-1]$, $\delta$ is a bell-shaped curve. Therefore, the minimum will be either at $2$ or $N-1$, values we can easily compute: $\delta(2)=(N-2)$ and $\delta(N-1)=\frac{(N-2)(N-1)}{2}$.

This function thus reaches its minimum for $k=2$, i.e. a link from $A_2$ to $A_1$, and takes value $\delta(2)=N-2$. Hence the best place to add the first link is as a backwards link to a link of the original DR.
\end{proof}

After this arc $2 \to 1$, the optimal addition of a second arc to a DR consists of another arc $3+k \to 2+k$, seeded following the same criteria as the first but it shall not be adjacent to the arc added in first place, i.e. $k>1$. That is, the two arcs shall not share a vertex. Indeed, the second arc again sets $d(3+k,2+k)=1$ and reduces the pathlength by $\Delta(P) = N-2$, but, if $k=1$, it reduces the pathlength further by setting $d(3,1)=2$. The non-adjacency condition is a very important observation. The same strategy will, however, no longer be valid for the addition of a third arc. In this case the optimal solution will be to organise the three arcs as a $3$-BS, rather than three non-adjacent $2$-BS.

\begin{lem}
Let $\tilde{G}$ be a directed ring of size $N$ and vertex ordering $v_1, v_2, \ldots v_N$. The addition of three arcs to $\tilde{G}$ forming a $3$-backwards subgraph decreases the total pathlength less than adding three non-adjacent $2$-backwards subgraphs.
\end{lem}

\begin{proof}
Let $\tilde{G}$ be a directed ring of size $N$. The goal is to identify the configuration of three additional arcs to $\tilde{G}$ such that the reduction in total pathlength $\Delta(P)$ is minimal. For that, we consider two cases and calculate the reduction to $P$ incurred by the addition of the arcs.

\emph{Case 1:} Consider the addition of three non-adjacent $2$-backwards subgraphs to $\tilde{G}$. Without loss of generality, consider the arcs $(v_2, v_1)$, $(v_4, v_3)$ and $(v_6, v_5)$. In the directed ring, to travel from $v_2$ to $v_1$ the whole ring has to be traversed. Thus, initially $d(v_2,v_1) = N-1$. After adding the arc $(v_2, v_1)$ the distance is now $d(v_2,v_1) = 1$, a reduction of $N-2$. The same happens for the two arcs $d(v_4,v_3)$ and $d(v_6,v_5)$ individually. Hence, the reduction to $P$ by the addition of the three arcs is $\Delta_1(P) = 3(N-2)$.
 
\emph{Case 2:} Assume the three arcs are added to the ring forming a $3$-BS, with arcs $\{ (v_2,v_1), (v_3,v_1), (v_3,v_2) \}$. As before, the distance from $v_2$ to $v_1$ decreases by $N-2$. The same happens for the distance $d(v_3,v_2)$. But the shortest path from $v_3$ to $v_1$ goes from $d(v_3,v_1) = N-2$ initially to $d(v_3,v_1) = N-2$. Altogether, the total pathlength is reduced by $\Delta_2(P) = 3N-7$. 

Concluding, since $\Delta_1(P) < \Delta_2(P)$, the configuration consisting of a $3$-BS is the one affecting less the path structure of the network.
\end{proof}

Similarly, one could show that adding one $4$-BS with 6 arcs decreases the pathlength of a directed ring by $6N-16$ while adding two non-adjacent $3$-BS (with 3 arcs each) reduces the total pathlength by $6N-14$. Hence, a directed ring with six extra arcs arranged into a $4$-BS is longer than the ring with two $3$-BS. We now generalise this result to arbitrary $M$.

\begin{lem}
Let $\tilde{G}$ be a directed ring of size $N$ and vertex ordering $v_1, v_2, \ldots v_N$. Let $M$ be an integer satisfying $M<N$. The addition of $L_M = \frac{1}{2}M(M-1)$ arcs forming an $M$-backwards subgraph to $\tilde{G}$ decreases the total pathlength $P$ of the ring by $\Delta(P|M)=\frac{1}{2} M(M-1) \, (N-\frac{M+4}{3})$. 
\end{lem}

\begin{proof}
Each vertex participating in an $M$-backwards subgraph sees the distance to all its predecessors become $d(v_i,v_j) = 1$. Thus, the first node sees no change, the second node sees one path going from $N-1$ to $1$, the third node sees two paths $N-1$ and $N-2$ to become $1$, and so on. The reduction of the total pathlength of the digraph $\Delta(P)$, with respect to the length of the initial ring given an $M$-component has been added can thus be written as $\Delta(P|M)=\sum_{j=1}^{M-1} \sum_{k=1}^j (N-j-1)$. This directly leads to the result.
\end{proof}

\begin{lem}
An $M$-backwards subgraph decreases the total pathlength per arc of a directed ring less than an $M'$-backwards subgraph if $M>M'$.
\end{lem}

\begin{proof}
As shown before, an $M$-backwards subgraph decreases the total length of a directed ring by $\Delta(P|M)=\frac{1}{2} M(M-1) (N-\frac{M+4}{3})$. We want to know in general what configuration provides the smallest decrease per link. To do so we normalise the decrease produced by an $M$-BS by its number of links $L_M = \frac{1}{2} M(M-1)$ in it:
	\begin{equation}
	 \Delta_n(M)=\frac{\frac{1}{2} M(M-1)(N-\frac{M+4}{3})} {\frac{1}{2}M(M-1)} = N - \frac{M+3}{3}.
	\end{equation}
The derivative of this expression respect to $M$ is simply $-1/3$, meaning that a bigger $M$-BS reduces the pathlength less a than smaller one per arc. 
\end{proof}

The last lemma has generalised the previous results, showing that it is always best to build a bigger $M$-BS. We now gather them into a Proposition.

\begin{prop} [Ultra-long digraphs \#1]  			\label{prop:ULdigraphs_1}
Let $\tilde{G}$ be a directed ring of size $N$ with labeled vertices and ordering $v_1, v_2, \ldots v_N$, and let $M$ be an integer satisfying $M \leq N$ and such that $\tilde{L}=N+\frac{M(M-1)}{2}$. Let $\tilde{G}_{UL}$ be the digraph resulting from adding the arc-set of an $M$-backwards subgraph to $\tilde{G}$. Then, the diameter, average pathlength and efficiency of $\tilde{G}_{UL}$ are given by:
	\begin{eqnarray}
		\tilde{\delta}_{UL} & = & N-1, \\
		\tilde{l}_{UL} & = & \frac{1}{2}N - \frac{1}{L_o} \frac{M\left( M-1 \right)}{2} \left[ N-\frac{M+4}{3} \right], \\
		\tilde{E}_{UL} & = & \frac{\psi(N)+\gamma}{N-1} + \frac{M-1}{\tilde{L}_o} \left[ \frac{M}{2} -\psi(N) \right] 
						 + \frac{1}{\tilde{L}_o}\sum_{j=1}^{M-1} \psi(N-j).
	\end{eqnarray}
Also, $\tilde{\delta}_{UL}$ is the longest diameter, $\tilde{l}_{UL}$ is the longest average pathlength and $\tilde{E}_{UL}$ is the smallest efficiency that a connected digraph with $\tilde{L} = N + \frac{1}{2} M(M-1)$ arcs can possibly have.
\end{prop}

\begin{proof}[Proof of Proposition~\ref{prop:ULdigraphs_1}]
The proof is a direct consequence of the previous lemmas.
\end{proof}

Note how similar this construction is to the connected ultra-long graphs. In the two cases, the base is the pathgraph or the DR and we build the largest possible fully connected subgraph with the extra links.
\begin{rem}
The expression for the average pathlength can be approximated to $\tilde{l}_{UL} \approx 2 + \rho - 2\sqrt{\rho} + \frac{N}{3} \big( 1 - 3 \rho + 2 \rho \sqrt{\rho} \, \big)$. 
\end{rem}

The result based on $M$-BS in Proposition~\ref{prop:ULdigraphs_1} works until $M = N-1$. From there, the next exact solution is slightly different because the arc $v_N \to v_1$, which should be part of the $N$-BS, already exists as part of the initial DR, see Fig 1(e) bottom leftmost. At this point,  $v_N$ is pointing to all its predecessors, the network contains $\tilde{L} = N -1 + \frac{1}{2} N(N-1)$ arcs and density is $\rho = \frac{1}{2} + \frac{1}{N}$.

\begin{defin} [Directed ring with full acyclic digraph]  			\label{def:ULdigraphs_DRN}
Let $N > 1$ be an arbitrary number of labeled vertices with an ordering $v_1, v_2, \ldots v_N$, and let $\tilde{L}_{DRN} = (N-1) + \frac{1}{2} N(N-1)$ be the desired number of arcs. Let $\mathcal{A}_C = \{ (v_i,v_{i+1}) : \; i = 1, 2,\ldots,(N-1) \}$ be the arc-set of a directed ring of size $N$, and let $\mathcal{A}_N$ be the arc-set of an $M$-backwards subgraph of order $M = N$. Then, a directed ring with full acyclic digraph $\tilde{G}_{DRN}$ consists in adding arc-sets $\mathcal{A}_C$ and $\mathcal{A}_N$ to the initially empty graph of size $N$.
\end{defin}

\begin{prop} [Ultra-long digraphs \#2]  			\label{prop:ULdigraphs_2}
Let $\tilde{G}(N,\tilde{L}_{DRN})$ be a directed ring with full acyclic digraph as in Definition~\ref{def:ULdigraphs_DRN}. Then, the diameter, average pathlength an efficiency of $\tilde{G}$ are given by:
	\begin{eqnarray}
		\tilde{\delta}_{UL} & = & N-1, 				\\
		\tilde{l}_{UL} & = & \frac{1}{6} (N + 4),			 \\
		\tilde{E}_{UL} & = & (N-1) \left( \frac{N}{2}+\gamma \right) + \sum_{i=1}^{N-1} \psi(i+1).
	\end{eqnarray}
\end{prop}

Finally, only the optimal strategy to generate ultra-long digraphs for the densest networks, when $\rho > \frac{1}{2} + \frac{1}{N}$ is left. Taking the $\tilde{G}_{DRN}$ special digraph in Definition~\ref{def:ULdigraphs_DRN} as the starting point for denser ultra-long digraphs, we bilateralise the remaining arcs in the same order they were originally added as part of subsequent $M$-backwards subgraphs. The difference is that optimal solutions are not restricted to specific groups which need to be included at a time, as it was the case for sparser digraphs and the $M$-BS solutions. In this case, for $\tilde{L} = \tilde{L}_{DRN}, \ldots, \tilde{L}_o$ the process is Markovian and arcs can be added individually to achieve an ultra-long digraph, see lower panel of Fig.~1(e) in main text.

\begin{defin} [Dense ultra-long digraphs]  			\label{def:ULdigraphs_Dense}
Let $\tilde{G}(N,\tilde{L}_{DRN})$ be a directed ring with full acyclic digraph as in Definition~\ref{def:ULdigraphs_DRN}. Let $\tilde{L}$ be the desired number of arcs satisfying $\tilde{L}_{DRN} < \tilde{L} < \tilde{L}_o$ and $\tilde{L}_r = \tilde{L} - \tilde{L}_{DRN}$ the number of remaining arcs. A dense ultra-long digraphs $\tilde{G}_{UL}$ is generated by orderly seeding the $\tilde{L}_r$ remaining arcs to $\tilde{G}(N,\tilde{L}_{DRN})$ in the forward direction such that every vertex receives arcs from its predecessors, one a time. That is $\{ (v_i, v_j)  : j = 3, 4, \ldots, N \textrm{ and } i = 1, \ldots, j-2 \}$ until $\tilde{L}_r$ are added. 
\end{defin}

Notice that in the previous definition index $i$ only runs until $j-2$ since the initial construction of a directed ring implies that arcs $(v_{j-1}, v_j)$ already exist.

\begin{prop} [Ultra-long digraphs \#3]  			\label{prop:ULdigraphs_3}
Let $\tilde{G}(N,\tilde{L})$ be a dense ultra-long digraph as in Definition~\ref{def:ULdigraphs_Dense} of size $N$ and $\tilde{L}$ arcs satisfying $N > 3$ and $\tilde{L}_{DRN} < \tilde{L} < \tilde{L}_o$ and $\tilde{L}_r = \tilde{L} - \tilde{L}_{DRN}$. Then, the average pathlength of $\tilde{G}$ is given by:
	\begin{eqnarray}
		\tilde{l}_{UL} & = & \frac{1}{6} (N + 4) + \frac{\tilde{L}_r}{N} - \frac{S(\tilde{L}_r)}{\tilde{L}_o},			 \\
	\end{eqnarray}
where $S(\cdot)$ is the reference $A060432$ in the On-Line Encyclopaedia of Integer Sequences (https://oeis.org/A060432). Also, $\tilde{l}_{UL}$ is the longest average pathlength that a connected digraph with $\tilde{L}$ arcs can possibly have.
\end{prop}
\begin{proof}[Proof of Proposition~\ref{prop:ULdigraphs_2}]
Once the directed ring with full acyclic graph is built, the arcs have to be bilateralised iteratively. It is easy to show that the optimal strategy is to bilateralise the arcs of the $M$-BS in the order they were built, i.e. take successively each node $3 \leq j \leq N$ along the directed ring and create an arc from each of its predecessors to this node $j$, see Fig.~1(e) of main text or Fig~\ref{fig:PerfectDAGs}. For each new node $3 \leq j \leq N$, $j-2$ arcs are added with the same contribution to the total pathlength: $N+1-j$. 

As far as the authors know, there is no closed-form formula for the impact on the pathlength. However the series consisting in summing the terms $1,2,2,3,3,3,4,4,4,4,...$ is referenced as $A060432$ in the On-Line Encyclopaedia of Integer Sequences (https://oeis.org/A060432). Let us name this sequence $S(n)$. Then, for a given number of links $\tilde{L} \geq N+\frac{N(N-1)}{2}-1$, we call $\tilde{L}_r=\tilde{L}-N-\frac{N(N-1)}{2}+1$. Then the total pathlength is $P = \frac{N(N-1) (4+N)}{6} + \tilde{L}_r(N-1) - S(\tilde{L}_r)$. Hence the result of the proposition.
\end{proof}
\begin{rem}
In the specific case where there exists an integer $K$ such that $\tilde{L}=\frac{N(N+1)}{2}-1+\frac{K(K-1)}{2}$, it is possible to completely bilateralise a $(K+1)$-BS and the pathlength reads:
    \begin{equation}
        \tilde{l}_{UL} = \frac{4+N}{6} - \frac{K(K-1)}{2 \tilde{L}_o} \left[ N-\frac{2(K+1)}{3} \right]
    \end{equation}
while the diameter is:
    \begin{equation}
        \tilde{\delta}_{UL} = N-K
    \end{equation}
\end{rem}

\begin{proof}
 The result is reached by grouping together the $j$ contributions of $N-j-1$ for $j$ running from $1$ to $K-1$.
 \begin{equation}
  \Delta = \frac{1}{\tilde{L}_o} \sum_{j=1}^{K-1} j(N-j-1)
 \end{equation}
\end{proof}

\subsection{Digraphs with smallest efficiency}

The last section dealt with the longest pathlength and smallest efficiency that strongly connected digraphs may achieve. 
As shown in the case of graphs, networks with the smallest efficiency are disconnected. The strategy followed to minimise efficiency consisted of maximising the number of disconnected nodes. In the cases when $L = \frac{1}{2} M(M-1)$, the graphs with smallest efficiency consisted of a complete graph of size $M$, leaving the remaining $N-M$ vertices isolated. The efficiency of such graphs equals the edge density: $E_{dUL} = \rho$, see Proposition~\ref{prop:dULgraphs_1}. This solution is still valid in the case of digraphs with the difference that $M$-complete subgraphs contain $\tilde{L} = M(M-1)$ directed arcs, since every edge is formed by two reciprocal arcs.

Because connectedness criteria in digraphs are more intricate than in graphs, we find other configurations giving rise to digraphs with smallest efficiency, and in particular with $\tilde{E} = \rho$. Taking advantage of the directionality of arcs, dense digraphs that are weakly connected and contain no cycles can be built. These are known as directed acyclic graphs (DAGs). We found a Markovian procedure to generate digraphs with $\tilde{E} = \rho$ for arbitrary values of $\tilde{L} \leq \tilde{L}_o / 2$, thus overcoming the limitation of the $M$-complete subgraph strategy which only holds only for specific values of $\tilde{L}$. This procedure is based on the  $M$-backwards subgraphs we employed to generate strongly connected ultra-long digraphs, but starting from an empty network instead of taking a directed ring as the baseline.

For $\tilde{L} > \tilde{L}_o / 2$ different strategies to obtain the smallest efficiency co-exist and exact solutions for all values of $\tilde{L}$ are very intricate. However, we remind that for the special cases when $\tilde{L} = M(M-1)$, the $M$-complete subgraph strategy is valid, giving $\tilde{E} = \rho$ at those values. Also, whenever $\tilde{L} = \frac{\tilde{L}_o}{2} + \frac{1}{2} M(M-1)$ we found that digraphs with $\tilde{E} = \rho$ can be constructed by filling the densest possible directed acyclic graph with arcs in the forward direction. Although these partial solutions leave some intermediate values of $\tilde{L}$ unexplained, they illustrate that considering $\tilde{E} = \rho$ as the smallest efficiency for a digraph of arbitrary number of arcs is a very reasonable assumption. Deviations from this limit for unexplored values of $\tilde{L}$ are expected to be small.

In the following we formalise these results. We start by the special cases defined by the $M$-complete subgraphs, which are inherited from the solutions for disconnected ultra-long graphs.

\begin{prop} [Disconnected ultra-long digraphs \#1]  		\label{prop:dULdigraphs_1}
Let $\tilde{G}_M(N)$ be a graph of size $N$ containing an $M$-complete subgraph satisfying $N>2$ and $M < N$. Then, $\tilde{G}_M(N)$ contains $L_M = M(M-1)$ directed arcs, all organised in reciprocal sets. The efficiency of the network is $\tilde{E}_{dUL} = \tilde{L} / \tilde{L}_o = \rho$, and $\tilde{E}_{dUL}$ is the smallest efficiency that any digraph of $N$ vertices and $\tilde{L} = \tilde{L}_M$ arcs can possibly have.
\end{prop}

\begin{proof}[Proof of Proposition~\ref{prop:dULdigraphs_1}]
As in previous demonstrations, we notice that, by construction, the distance matrix of $\tilde{G}_M(N)$ has $n_1 = \tilde{L}$ entries with  $d_{ij} = 1$ and all other entries take an infinite value. It is then trivial to prove that its efficiency equals the arc density and that this is the smallest efficiency a digraph of size $N$ and $\tilde{L}$ arcs could possibly have.
\end{proof}

\begin{figure*}[] 
	\centering
 	\includegraphics[width=1.0\textwidth,clip=]{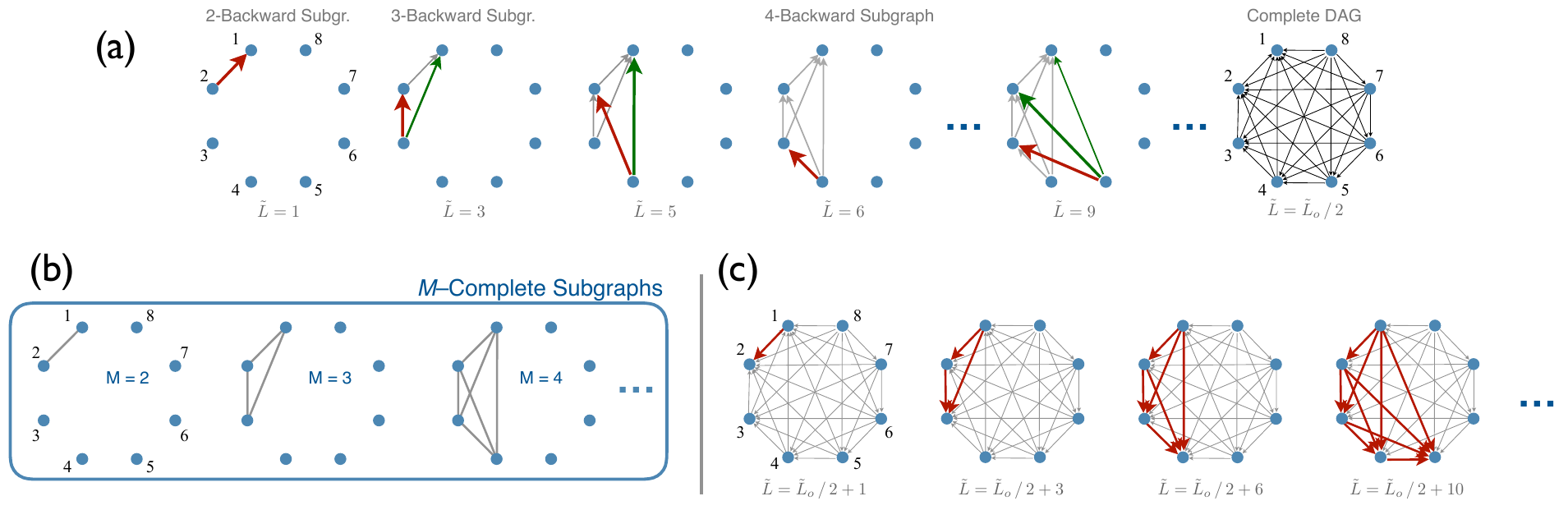}
 	\caption{ 	\label{fig:PerfectDAGs}
	{\bf Construction of disconnected ultra-long digraphs.} 
	(a) For $\tilde{L} < \tilde{L}_o / 2$, the configuration minimising efficiency consist of a directed acyclic graph (DAG) in which the distance between all connected nodes is $d=1$. This is achieved by constructing consecutive $M$-backwards subgraphs on an initially empty network. When $\tilde{L} = \tilde{L}_o / 2$ the model leads to a complete DAG, which is the densest possible directed acyclic graph.
	(b) Previously identified $M$-Complete subgraphs are also special solutions to generate (di)graphs with smallest efficiency.
	(c) For $\tilde{L} > \tilde{L}_o / 2$ arcs need to be added in the opposite orientation giving rise to cycles. Although several configurations may exist, for the special cases in which $\tilde{L}_M = \frac{1}{2}M(M-1)$ the introduction of $M$-forward subgraphs creates a complete graph within the network. Thus, all cycles in the network are  reciprocal connections between pairs of nodes.
	} 
\end{figure*}

We now show that a Markovian addition of arcs, consisting of a smooth transition between $M$-backwards subgraphs of consecutive $M$, leads to digraphs with smallest efficiency possible for all $\tilde{L} \leq \tilde{L}_o / 2$.

\begin{defin} [Ultra-long directed acyclic graph]  		\label{def:dUL_DAG}
Let $N$ and $\tilde{L}$ be arbitrary numbers of vertices and directed arcs satisfying $N \geq 3$ and $\tilde{L} \leq \tilde{L}_o / 2$. Let $v_1, v_2, \ldots v_N$ be an ordering of the vertices. Let 
	\begin{equation}
	M = \left\lfloor \frac{1}{2} \left( 1+\sqrt{1 + 8\tilde{L})} \right) \right\rfloor 	\nonumber
	\end{equation}
be the size of the largest $M$-backwards subgraph which can be constructed with $\tilde{L}$ arcs. Such backwards subgraph contains $\tilde{L}_M = \frac{1}{2}M(M-1)$ arcs. Finally, let $\tilde{L}_e = \tilde{L} - \tilde{L}_M$ be the number of excess arcs. Then,
\begin{itemize}
\item If $\tilde{L}_e = 0$, an ultra-long directed acyclic graph $\tilde{G}(N,\tilde{L})$ is made of an $M$-backwards subgraph and $(N-M)$ isolated vertices.
\item If $\tilde{L}_e > 0$, $\tilde{G}(N,\tilde{L})$ consists of an $M$-backwards subgraph and the remaining arcs are placed pointing from vertex $i = M+1$ to the first $\tilde{L}_e$ nodes.
\end{itemize}
\end{defin}
 
\begin{defin} [Complete directed acyclic graph]  		\label{def:CompleteDAG}
A complete directed acyclic graph $\tilde{K}_N$ is the ultra-long DAG of size $N$ and order $M = N$.
\end{defin}

\begin{rem} [Properties of complete DAGs]  			\label{rem:CompleteDAG}
Given a complete DAG $\tilde{K}_N$ of size $N$, then:
 	\begin{itemize}
	\item[--] The vertex $v_i$ with ordering $i$ has in-degree $k_i^- = N-i$ and out-degree $k_i^+ = i-1$.
	\item[--] There is only one vertex $(v_1)$ with out-degree zero and only one vertex $(v_N)$ with in-degree zero.
	\item[--] The pathlength between two vertices in $\tilde{K}_N$ is $d_{ij} = 1$ if $i < j$ and $d_{ij} = \infty$ if $i \geq j$.
	\item[--] The density of $\tilde{K}_N$ is $\rho = \frac{1}{2}$ and its efficiency is $\tilde{E}(\tilde{K}_N) = \frac{1}{2}$, as well.
	\end{itemize}
\end{rem}

\begin{prop} [Disconnected ultra-long digraphs \#2]  \label{prop:dULdigraphs_2}
Let $\tilde{G}(N,\tilde{L})$ be an ultra-long DAG of $N$ vertices and $\tilde{L}$ arcs, where $\tilde{L} \leq \frac{1}{2}\tilde{L}_o$. The efficiency of $\tilde{G}(N,L)$ equals its link density, $\tilde{E}_{dUL} = \tilde{L} / \tilde{L}_o = \rho$, and $\tilde{E}_{dUL}$ is the smallest efficiency any digraph of $N$ vertices and $\tilde{L}$ arcs can possibly have.
\end{prop}

\begin{proof}[Proof of Proposition~\ref{prop:dULdigraphs_2} ]
By construction, the pathlength between two vertices in an ultra-long DAG is $d_{ij} = 1$ if the arc $(v_i,v_j)$ exists and otherwise $d_{ij} = \infty$. Hence the pairwise distance matrix (ignoring diagonal entries) contains $n_1 = \tilde{L}$ entries with $d_{ij} = 1$ and $n_\infty = \tilde{L}_o - \tilde{L}$ entries with $d_{ij} = \infty$. 
Following previous demonstrations, e.g., proof of Proposition~\ref{prop:dULgraphs_1}, it is trivial to show that an ultra-long DAG gives rise to the smallest possible efficiency that a network of size $N$ and $\tilde{L}$ arcs could possibly have.
\end{proof}

Finally, for $\tilde{L} > \tilde{L}_o / 2$ we have found again that different configurations may compete for the digraph with lowest efficiency. Whenever $\tilde{L} = M(M-1)$ the $M$-complete configuration and Proposition~\ref{prop:dULdigraphs_1} are still valid in this range. We find yet another set of special cases for which $\tilde{E} = \rho$. These consists of adding $\frac{1}{2} M(M-1)$ arcs to a complete DAG such that all relations between the first $M$ vertices are bilateralised at once, in the same spirit as the $M$-backwards subgraphs were added to generate connected ultra-long digraphs, but now with the arcs in the forward direction, Fig.~\ref{fig:PerfectDAGs}(bottom).

\begin{defin} [$M$-forward subgraph]  			\label{def:MFS}
Let $\tilde{G}$ be a digraph of size $N$ with labeled vertices and ordering $v_1, v_2, \ldots v_N$. Let $M$ be an integer satisfying $M \leq N$. Then, an $M$-forward subgraph (or simply $M$-FS) consists of a subset of $M$ consecutive vertices $V_M = \{v_{k}, v_{k+1}, \ldots, v_{k+M} \}$ of $\tilde{G}$, with each vertex pointing to all its successors within the set. The arc-set of an $M$-FS is thus $A_M =\{(v_i,v_j): \,i = k, \ldots, k+M$ \textrm{and} $i < j \}$. An $M$-FS contains $L_M = \frac{1}{2} M(M-1)$ arcs.
Unless otherwise stated, it shall be understood that $k=1$ and $v_k$ is the first node in the ordering of the digraph.
\end{defin}

\begin{prop} [Disconnected ultra-long digraphs \#3]  \label{prop:dULdigraphs_3}
Let $\tilde{K}_N$ be a complete DAG of size $N$ as in Definition~\ref{def:CompleteDAG}. Let $\tilde{G}(N)$ be the graph resulting from the union of $\tilde{K}_N$ and an $M$-forward subgraph with $2 < M \leq N$. Then $\tilde{G}(N)$ contains $\tilde{L} = \frac{\tilde{L}_o}{2} + \frac{1}{2} M(M-1)$ arcs, its efficiency is $\tilde{E}_{dUL} = \tilde{L} / \tilde{L}_o = \rho$, and $\tilde{E}_{dUL}$ is the smallest efficiency any digraph of $N$ vertices and $\tilde{L}$ arcs can possibly have.
\end{prop}

\begin{proof}[Proof of Proposition~\ref{prop:dULdigraphs_3} ]
By construction, the distance matrix of $\tilde{G}$ has $n_1 = \tilde{L}$ entries with  $d_{ij} = 1$ and all other entries take an infinite value. As in previous demonstrations, it is trivial to prove that the efficiency of the digraph equals its density and that this is the smallest efficiency a digraph of size $N$ and $\tilde{L}$ arcs could possibly have.
\end{proof}

\section{Numerical search for ultra-long digraphs}		\label{sec:UL_BruteForce}

Exact identification of the extremal network configurations turns very challenging in some cases since different configurations may exist or even co-exist depending on the precise number of links. Therefore, we have performed exhaustive numerical searches with networks of small size to clarify those cases. In the following we illustrate these efforts in the cases of disconnected ultra-long graphs and ultra-long digraphs.

\subsection{Disconnected ultra-long graphs}		\label{sec:BruteForce_dULgraphs}

In order to systematically search for all possible graphs with smallest efficiency, we started by identifying all non-isomorphic undirected graphs $G(N,L)$ of sizes $N=5$ to $10$ using the software \texttt{nauty}~\cite{McKay_NautySoftware_2013} for each number of edges running from $L =1$ to $L_o = \frac{1}{2} N(N-1)$. The efficiency for all non-isomorphic graphs of a given $L$ was calculated and those with the smallest value were conserved. The results are summarised in Figures~\ref{fig:dULgraphs_1} -- \ref{fig:dULgraphs_3}. As seen, for each $L$ usually more than one configuration exists which leads to the smallest efficiency. However, for the particular cases in which $L = \frac{1}{2}M(M-1)$ with $M = 1, 2, \ldots, N$ the $M$-complete disconnected graphs, introduced in Definition~\ref{def:M_dULgraph}, are always a solution.

We have also compared the empirically obtained efficiency with the limiting value $E_{dUL} = L / L_o = \rho$. The results for graphs of sizes $N = 5,6, 9$ and $10$ are shown in Fig.~\ref{fig:Effic_dULgraphs}. As seen, deviations from density when $L < L_o / 2$ are rare and marginal. For denser networks, $L > L_o / 2$, deviations from density happen more often and they are more prominent. However, the magnitude of the deviations decreases for larger networks. These results evidence that considering the lowest boundary of efficiency for graphs as $E_{dUL} = \rho$ is a very reasonable assumption.
The point of largest deviation between empirical efficiency and $E_{dUL} = \rho$ happens when $L^* = \frac{1}{2} (N-1)(N-2) + 1$. We have shown in Sec.~\ref{sec:dULgraphs} that this maximal difference rapidly decays with network size and falls below the $1\%$ relative error whenever $N > 100$.

\begin{figure} 
	\centering
 	\includegraphics[width=0.85\textwidth,clip=]{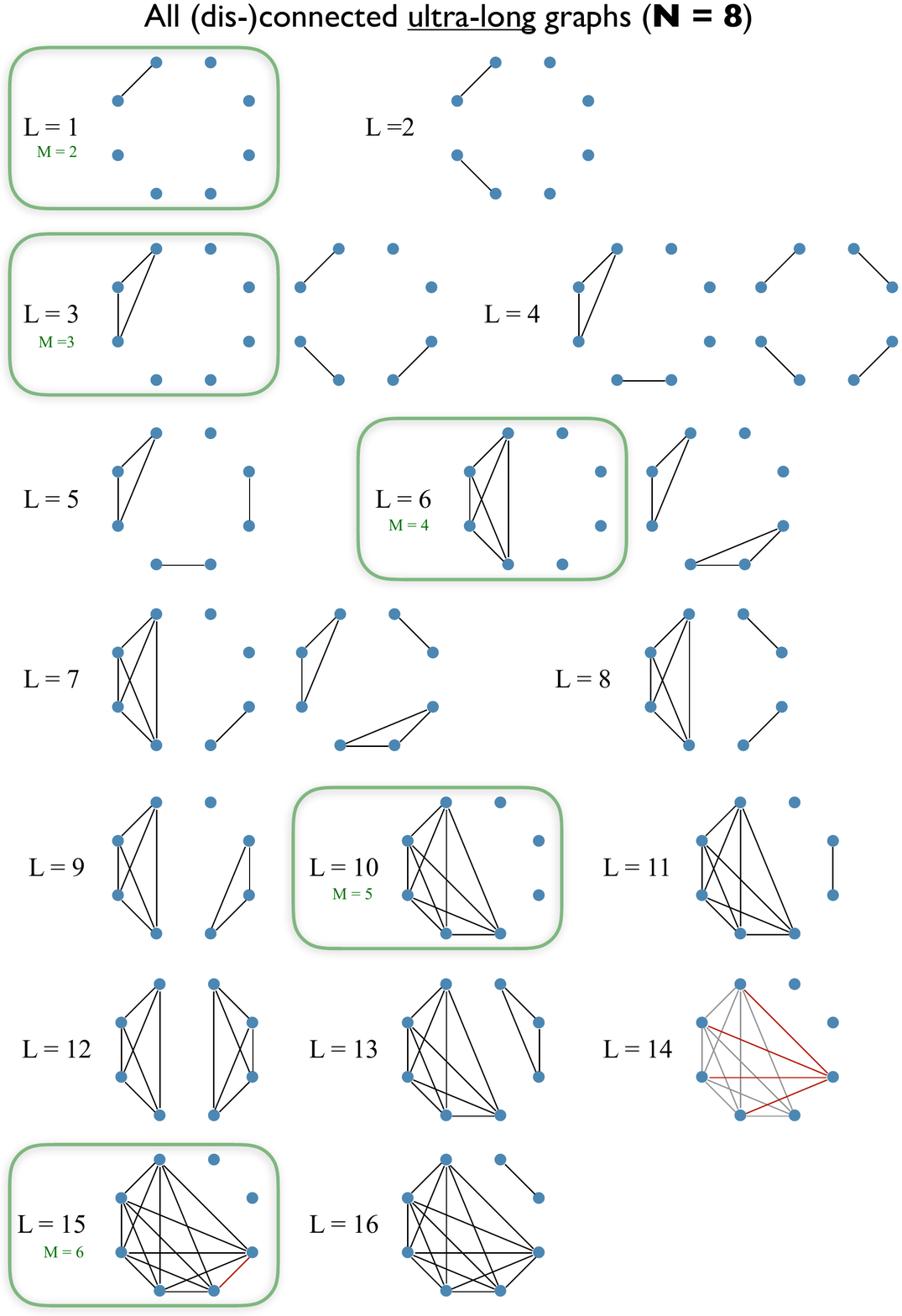}
 	\caption{ \label{fig:dULgraphs_1}
	{\bf Disconnected ultra-long graphs.} Collection of all non-isomorphic configurations of graphs with smallest possible efficiency for graphs of size $N = 8$ and arbitrary number of edges. While in general several configurations exists, for the cases when $L_M = \frac{1}{2} M(M-1)$, an UL graph always exists which is made of a complete subgraph of size $M$ and $(N-M)$ isolated nodes. In these cases, $E_{dUL} = \rho$. Red edges mark the last arc(s) seeded in the cases when an dUL graph is represented by the Markovian generative method.
	}
\end{figure}

\begin{figure} 
	\centering
 	\includegraphics[width=0.85\textwidth,clip=]{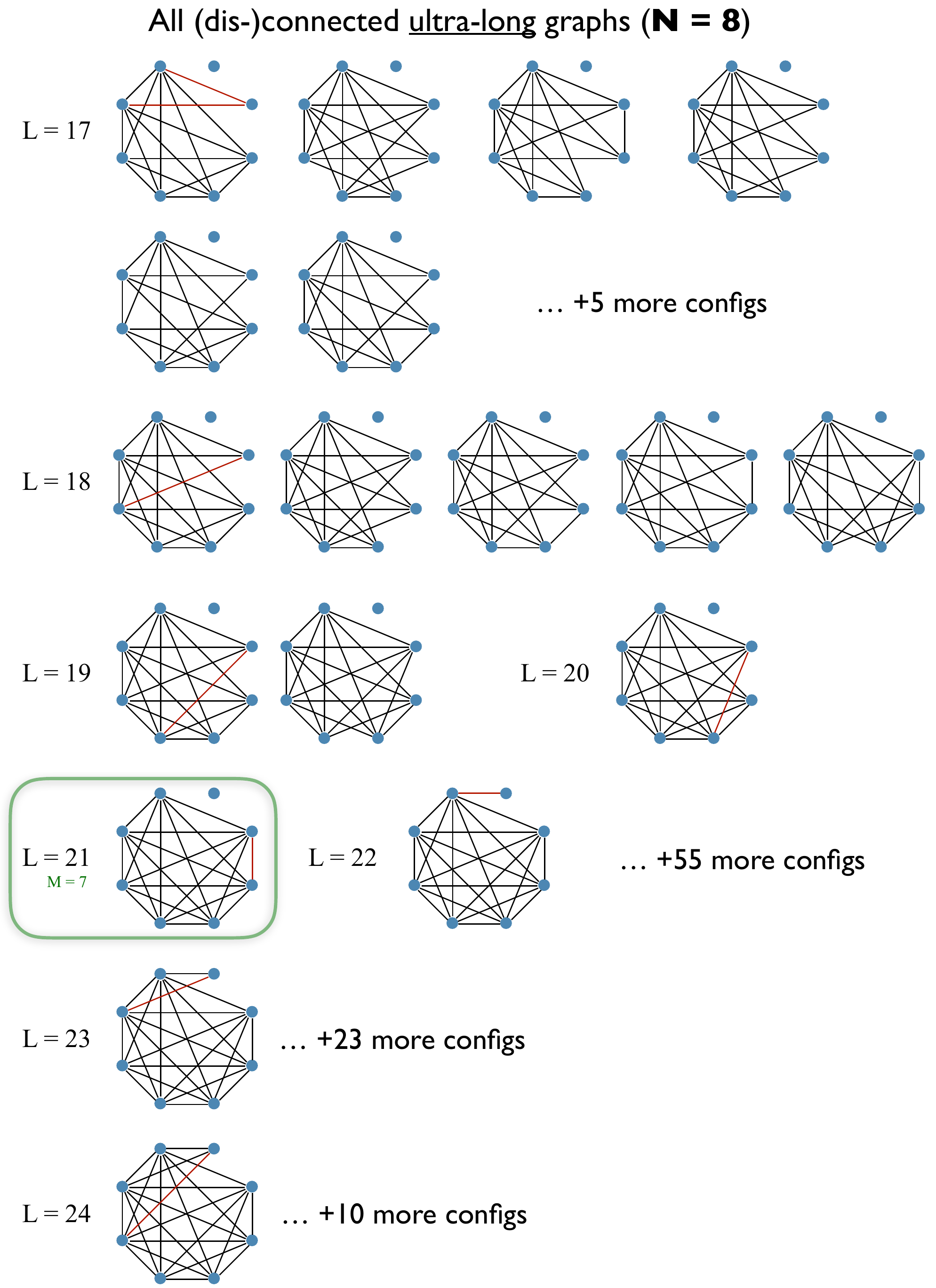}
 	\caption{ \label{fig:dULgraphs_2}
	{\bf Disconnected ultra-long graphs.} Collection of all non-isomorphic configurations of graphs with smallest possible efficiency for graphs of size $N = 8$ and arbitrary number of edges. While in general several configurations exists, for the cases when $L_M = \frac{1}{2} M(M-1)$ edges, an UL graph always exists which is made of a complete subgraph of size $M$ and $(N-M)$ isolated nodes. In these cases, $E_{dUL} = \rho$. Red edges mark the last arc(s) seeded in the cases when an dUL graph is represented by the Markovian generative method.
	}
\end{figure}

\begin{figure} 
	\centering
 	\includegraphics[width=0.85\textwidth,clip=]{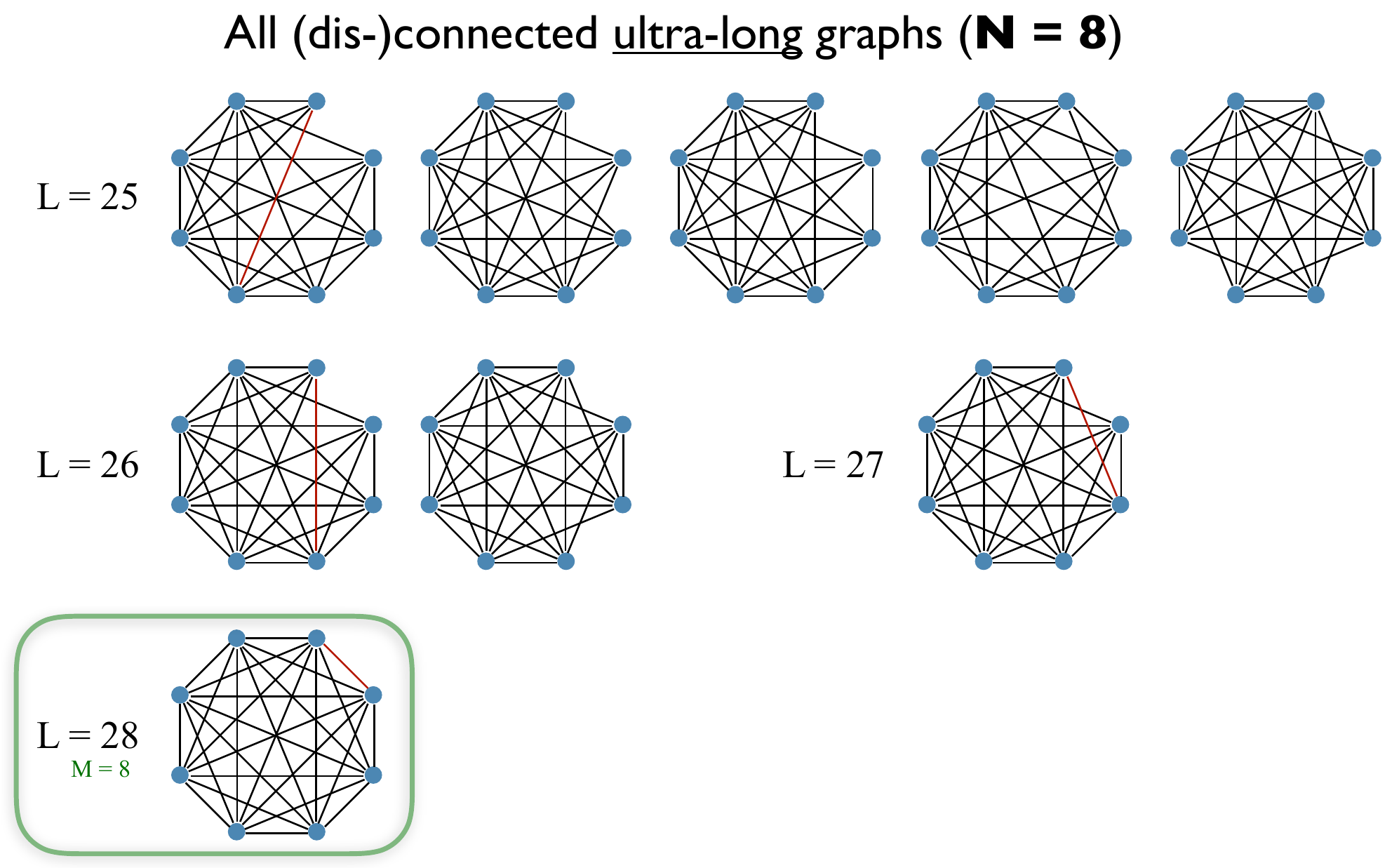}
 	\caption{ \label{fig:dULgraphs_3}
	{\bf Disconnected ultra-long graphs.} Collection of all non-isomorphic configurations of graphs with smallest possible efficiency for graphs of size $N = 8$ and arbitrary number of edges. While in general several configurations exists, for the cases when $L_M = \frac{1}{2} M(M-1)$ edges, an UL graph always exists which is made of a complete subgraph of size $M$ and $(N-M)$ isolated nodes. In these cases, $E_{dUL} = \rho$. Red edges mark the last arc(s) seeded in the cases when an dUL graph is represented by the Markovian generative method.
	}
\end{figure}

\begin{figure} 
	\centering
 	\includegraphics[width=0.45\textwidth,clip=]{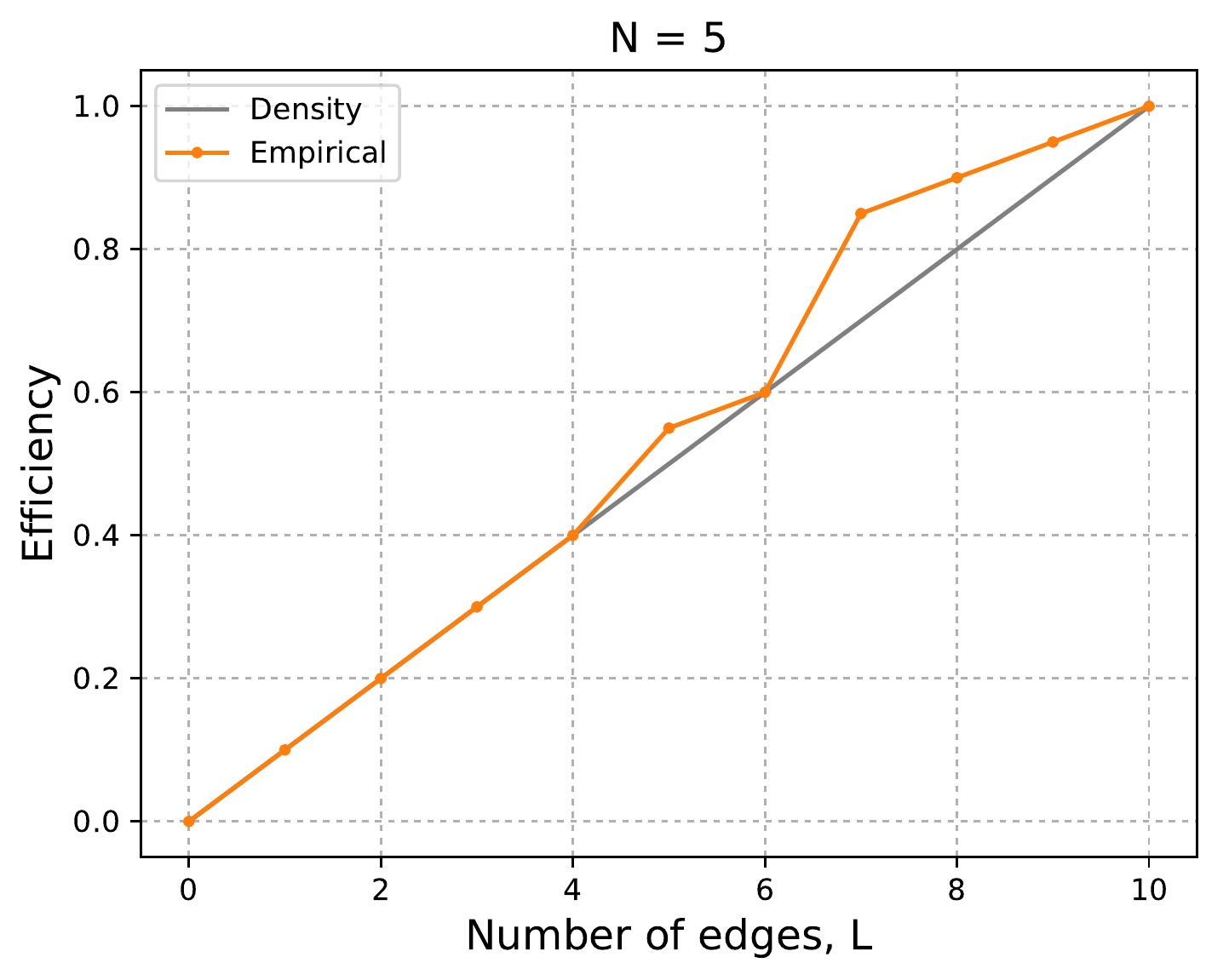}
 	\includegraphics[width=0.45\textwidth,clip=]{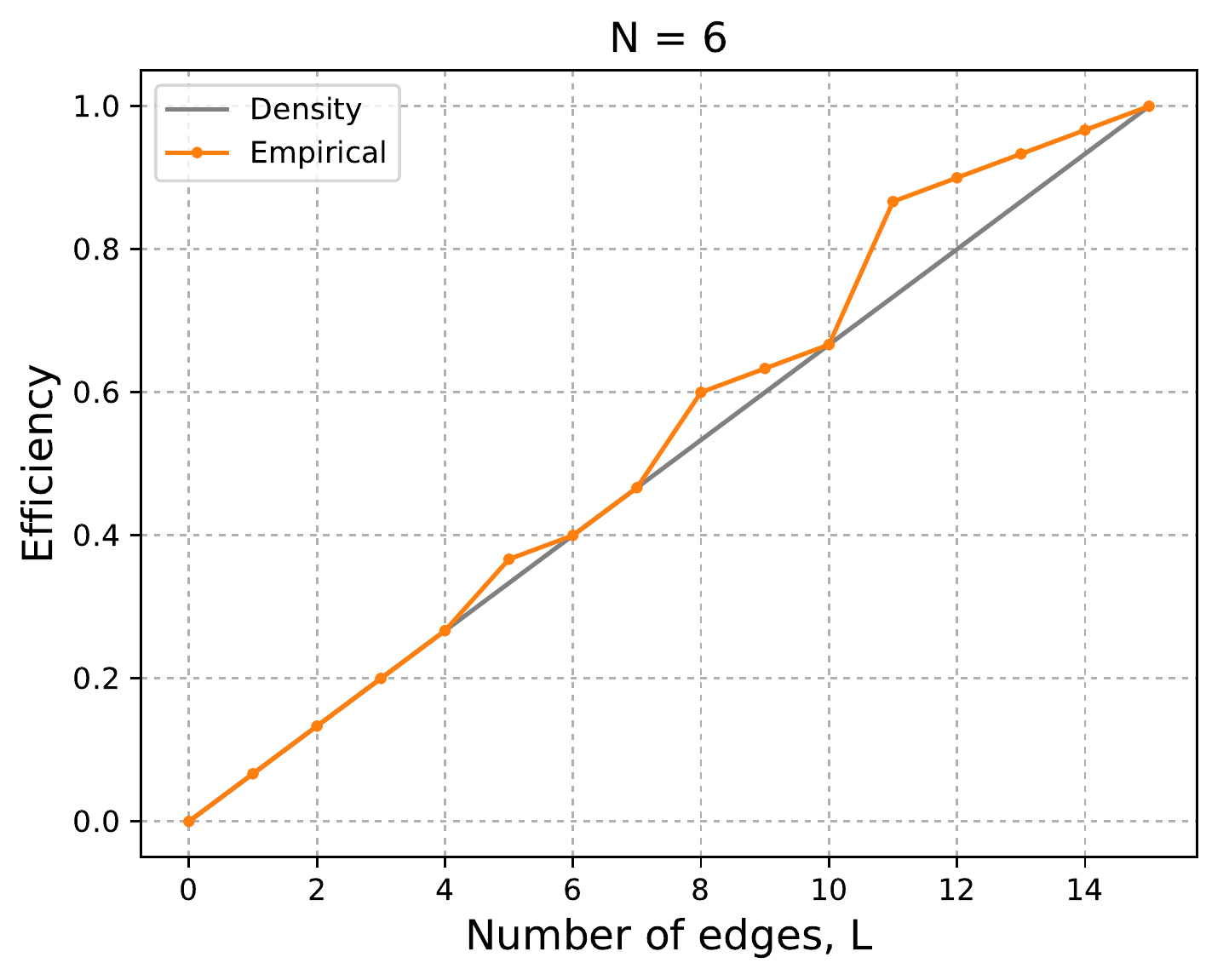}
 	\includegraphics[width=0.45\textwidth,clip=]{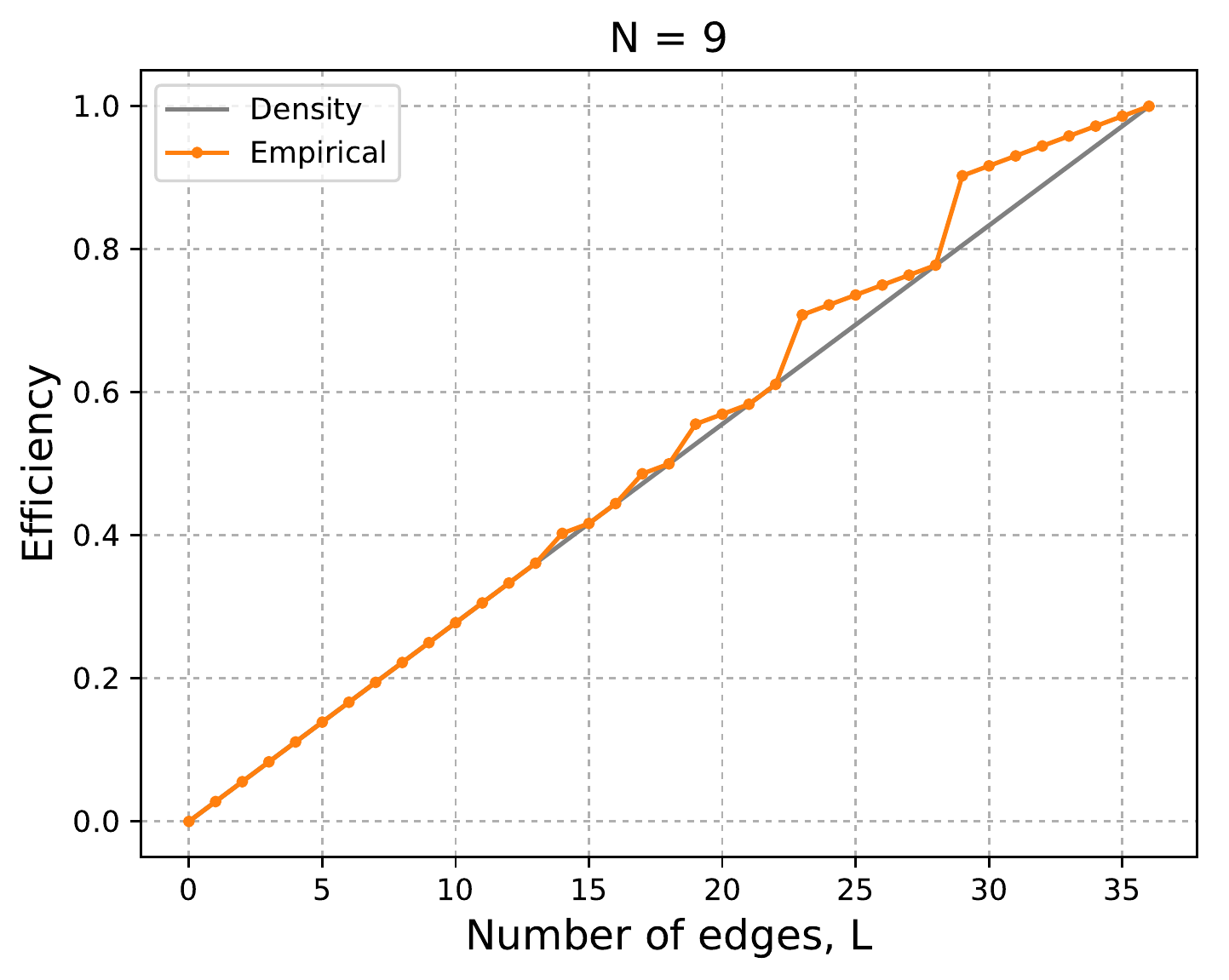}
 	\includegraphics[width=0.45\textwidth,clip=]{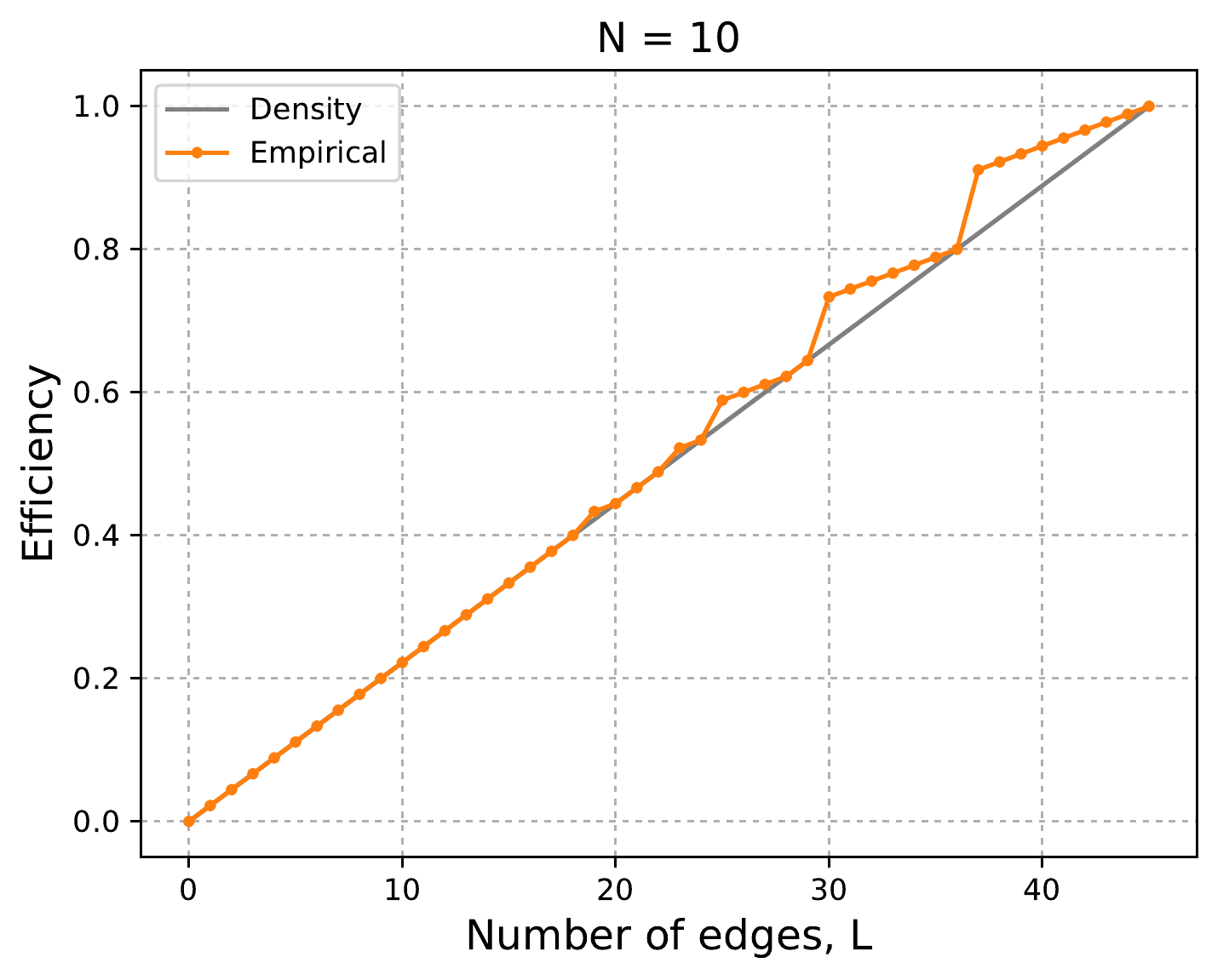}
 	\caption{ \label{fig:Effic_dULgraphs}
	{\bf Efficiency of disconnected ultra-long graphs} for networks of sizes $N = 5, 6, 9$ and $10$ at all densities, from $L = 0$ to $L = L_o = \frac{1}{2} N(N-1)$ edges. Empirically identified smallest efficiency from an exhaustive search of all existing non-isomorphic graphs is shown in orange, precise values marked with dots. 
	The density of the graphs is shown in grey, which represents an exact solution for the efficiency for most of values of $L$ and is an excellent approximation in others. Largest deviation of the empirical efficiency happens at $L^* = \frac{1}{2} (N-1)(N-2) + 1$. At this high density the graph necessarily becomes connected and forms a $(N,N\!-\!1)$-kite graph, with the last remaining node being connected to one of the nodes in the $(N\!-\!1)$-Complete subgraph.
	}
\end{figure}

\subsection{Ultra-long digraphs}		\label{sec:BruteForce_ULdigraphs}

Here we present all configurations leading to connected digraphs with longest possible pathlength. We started by identifying all non-isomorphic undirected graphs $G(N,L)$ of size $N=5$ and $L \in [1, L_o]$ using the software \texttt{nauty}~\cite{McKay_NautySoftware_2013}. Out of each identified $G(N,L)$, we extracted all possible labeled digraphs embedded in $G$, for all $\tilde{L} \in [1,L]$, and kept only the non-isomorphic set using the \texttt{iGraph} software (python-iGraph 0.7.0, www.igraph.org). Once all non-isomorphic digraphs $\tilde{G}(N,\tilde{L})$ of $N=5$ vertices had been identified for all $\tilde{L} \in [1, \tilde{L}_o]$, their average pathlength was computed (for all strongly connected configurations) and the ones maximising the pathlength were conserved.

The results are summarised in Figures~\ref{fig:AllULdigraphsN5_p1} and~\ref{fig:AllULdigraphsN5_p2}. As expected, there is in general a variety of configurations leading to the longest average pathlength for a given number of arcs. Most of the combinations seem unrelated making the definition and algorithmic generation of connected UL digraphs very challenging. However, as we predicted, for the cases where $\tilde{L} = N + \frac{1}{2}M(M-1)$ with $M = 1,2,\ldots, (N+1)$ there exist a unique UL digraph, consisting of the superposition of a directed ring and an $M$-backwards subgraph. These special cases allow for the existence of one Markovian path to generate UL digraphs of arbitrary number of arcs, despite the variety of configurations occurring for given values of $\tilde{L}$. In Figs.~\ref{fig:AllULdigraphsN5_p1} and~\ref{fig:AllULdigraphsN5_p2} this Markovian path is highlighted by the green arrows, signalling the arc(s) added to an existing UL digraph of $\tilde{L}$ arcs leading to a new UL digraph with $\tilde{L} +1$ arcs.

\begin{figure} 
	\centering
 	\includegraphics[width=1.0\textwidth,clip=]{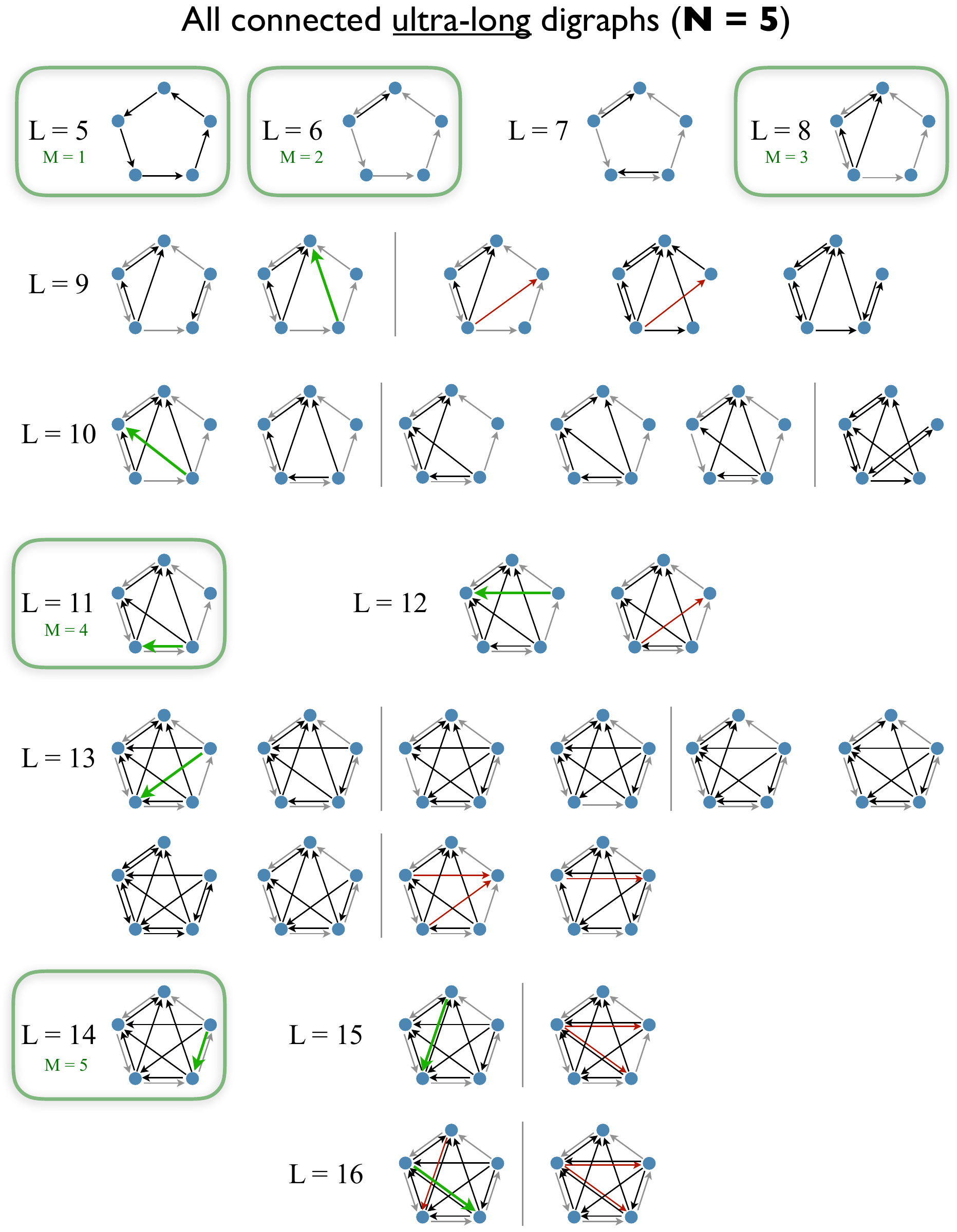}
 	\caption{ \label{fig:AllULdigraphsN5_p1}
	{\bf Connected ultra-long digraphs.} Collection of all existing (non-isomorphic) connected ultra-long digraphs of size $N=5$ and arbitrary number of arcs. While in general several configurations exists, for the cases with $\tilde{L} = N + \frac{1}{2} M(M-1)$ arcs, UL digraphs are unique and consist of a directed ring with an $M$-backwards subgraph superimposed. Green arrows highlight a Markovian path to generate at least one UL digraph. Red arrows mark the arcs seeded in opposite orientation to the initial directed ring.
	}
\end{figure}

\begin{figure} 
	\centering
 	\includegraphics[width=1.0\textwidth,clip=]{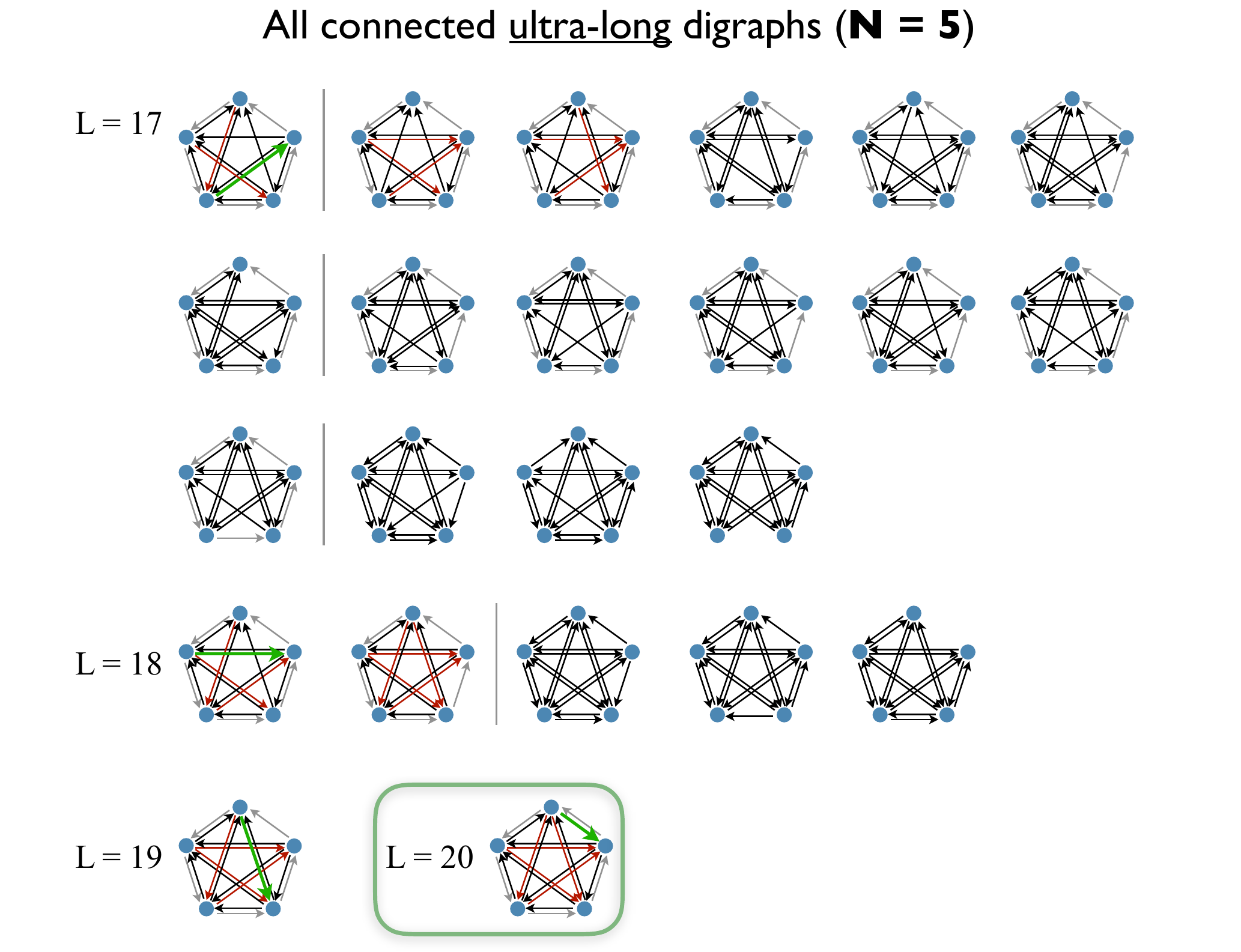}
 	\caption{ \label{fig:AllULdigraphsN5_p2}
	{\bf Connected ultra-long digraphs (continued).} Collection of all existing (non-isomorphic) connected ultra-long digraphs of size $N=5$ and arbitrary number of arcs. While in general several configurations exists, for the cases with $\tilde{L} = N + \frac{1}{2} M(M-1)$ arcs, UL digraphs are unique and consist of a directed ring with an $M$-backwards subgraph superimposed. Green arrows highlight a Markovian path to generate at least one UL digraph. Red arrows mark the arcs seeded in opposite orientation to the initial directed ring.
	}
\end{figure}


\end{document}